\ifdefined\IsLLNCS
\documentclass{llncs}
	\usepackage{amsfonts,amsmath,amssymb,boxedminipage,color,url,nccmath}
\else
	\documentclass[11pt]{article}
	\usepackage{amsfonts,amsmath,amssymb,boxedminipage,color,url,fullpage,nccmath,amsthm}
\fi
\usepackage{color}
\usepackage{framed}

\definecolor{weborange}{rgb}{.8,.3,.3}
\definecolor{webblue}{rgb}{0,0,.8}
\definecolor{internallinkcolor}{rgb}{0,.5,0}
\definecolor{externallinkcolor}{rgb}{0,0,.5}

\usepackage[pagebackref,
hyperfootnotes=false,
colorlinks=true,
urlcolor=externallinkcolor,
linkcolor=internallinkcolor,
filecolor=externallinkcolor,
citecolor=internallinkcolor,
breaklinks=true,
pdfstartview=FitH,
pdfpagelayout=OneColumn]{hyperref}

\usepackage{enumerate,paralist}
\usepackage[numbers]{natbib} 
\usepackage{aliascnt}
\usepackage{cleveref}
\usepackage{xspace}

\usepackage{relsize}
\usepackage{graphics}
\ifdefined\IsLLNCS
\fi

\providecommand{\remove}[1]{}

\newcommand{\Draft}[1]{\ifdefined\IsDraft \texttt{ #1} \fi}
\newcommand{\LLNCS}[1]{\ifdefined\IsLLNCS #1 \fi}

\newcommand{\TLLNCS}[2]{\ifdefined\IsLLNCS#1\else #2 \fi}

\newcommand{\sdotfill}{\textcolor[rgb]{0.8,0.8,0.8}{\dotfill}} 

\newenvironment{protocol}{\begin{proto}}{\vspace{-\topsep}\sdotfill\end{proto}}
\newenvironment{algorithm}{\begin{algo}}{\vspace{-\topsep}\sdotfill\end{algo}}

\newcommand{\ie} {i.e.,\ }
\newcommand{\eg} {e.g.,\ }
\newcommand{\wrt} {with respect to\ }
\newcommand{\wlg} {without loss of generality\xspace}

\renewcommand{\vec}[1]{\textbf{#1}}

\newcommand{\ceil}[1]{\left\lceil #1 \right\rceil}

\newcommand{\set}[1]{\ens{#1}}

\newcommand{\floor}[1]{\left \lfloor#1 \right \rfloor}

\newcommand{\eqdef}{:=}

\newcommand{\N}{{\mathbb{N}}}

\newcommand{\zo}{\set{0,1}}
\newcommand{\zn}{\zo^n}
\newcommand{\zs}{\zo^\ast}

\newcommand{\suchthat}{{\;\; : \;\;}}

\newcommand{\getsr}{\mathbin{\stackrel{\mbox{\tiny R}}{\gets}}}
\newcommand{\xor}{\oplus}

\newcommand{\eps}{\varepsilon}

\newcommand{\Exp}{\operatorname*{E}}

\newcommand{\Supp}{\operatorname{Supp}}

\newcommand{\MathFam}[1]{\mathcal{#1}}

\newcommand{\FFam}{\MathFam{F}}

\newcommand{\p}{{{\mathsf{P}}}}

\renewcommand{\cref}{\Cref}

\TLLNCS{
	\newaliascnt{claiml}{theorem}
	\newtheorem{claiml}[claiml]{Claim}
	\aliascntresetthe{claiml}
	
	\renewenvironment{claim}{\begin{claiml}}{\end{claiml}}
}
{
	\newtheorem{theorem}{Theorem}[section]

	\newaliascnt{lemma}{theorem}
	\newtheorem{lemma}[lemma]{Lemma}
	\aliascntresetthe{lemma}
	
	\newaliascnt{claim}{theorem}
	\newtheorem{claim}[claim]{Claim}
	\aliascntresetthe{claim}

	\newaliascnt{corollary}{theorem}
	\newtheorem{corollary}[corollary]{Corollary}
	\aliascntresetthe{corollary}
	
	\newaliascnt{proposition}{theorem}
	
	\aliascntresetthe{proposition}

	\newaliascnt{conjecture}{theorem}
	
	\aliascntresetthe{conjecture}

	\newaliascnt{definition}{theorem}
	\newtheorem{definition}[definition]{Definition}
	\aliascntresetthe{definition}

	\newaliascnt{remark}{theorem}
	
	\aliascntresetthe{remark}

	\newaliascnt{example}{theorem}
	
	\aliascntresetthe{example}
}

\crefname{lemma}{Lemma}{Lemmas}
\crefname{figure}{Figure}{Figures}
\crefname{claim}{Claim}{Claims}
\crefname{corollary}{Corollary}{Corollaries}
\crefname{proposition}{Proposition}{Propositions}
\crefname{conjecture}{Conjecture}{Conjectures}
\crefname{definition}{Definition}{Definitions}
\crefname{remark}{Remark}{Remarks}
\crefname{exmaple}{Example}{Examples}

\newaliascnt{construction}{theorem}

\aliascntresetthe{construction}
\crefname{construction}{Construction}{Constructions}

\newaliascnt{fact}{theorem}
\newtheorem{fact}[fact]{Fact}
\aliascntresetthe{fact}
\crefname{fact}{Fact}{Facts}

\newaliascnt{notation}{theorem}

\aliascntresetthe{notation}
\crefname{notation}{Notation}{Notation}

\crefname{equation}{Equation}{Equations}

\newaliascnt{proto}{theorem}

\newtheorem{proto}[proto]{Protocol}

\aliascntresetthe{proto}
\crefname{proto}{protocol}{protocols}

\newaliascnt{algo}{theorem}
\newtheorem{algo}[algo]{Algorithm}
\aliascntresetthe{algo}
\crefname{algo}{algorithm}{algorithms}

\newaliascnt{expr}{theorem}
\newtheorem{expr}[expr]{Experiment}
\aliascntresetthe{expr}
\crefname{experiment}{experiment}{experiments}

\newcommand{\stepref}[1]{Step~\ref{#1}}

\def\FullBox{$\Box$}
\def\qed{\ifmmode\qquad\FullBox\else{\unskip\nobreak\hfil
\penalty50\hskip1em\null\nobreak\hfil\FullBox
\parfillskip=0pt\finalhyphendemerits=0\endgraf}\fi}

\def\qedsketch{\ifmmode\Box\else{\unskip\nobreak\hfil
\penalty50\hskip1em\null\nobreak\hfil$\Box$
\parfillskip=0pt\finalhyphendemerits=0\endgraf}\fi}

\newcommand{\ex}[2]{\Exp_{#1}\left[#2\right]}
\newcommand{\pr}[1]{\Pr\left[#1\right]}
\newcommand{\ppr}[2]{\Pr_{#1}\left[#2\right]}

\newcommand{\Sa}{\mathsf{S}}

\newcommand{\ens}[1]{\left\{#1\right\}}
\newcommand{\size}[1]{\left|#1\right|}

\newcommand{\out}{\operatorname{out}}

\newcommand{\trans}{{\sf trans}}

\newcommand{\prob}[1]{\mathsf{\textsc{#1}}}

\newcommand{\SD}{\prob{SD}}

\newcommand{\cs}{{\cal{S}}}
\newcommand{\cu}{{\cal{U}}}

\newcommand{\q}{{\cal{Q}}}

\newcommand{\cE}{\mathcal{E}}

\newcommand{\st}{\suchthat}

\newcommand{\Ac}{\mathsf{A}}
\newcommand{\Bc}{\mathsf{B}}

\newcommand{\Ec}{\mathsf{E}}

\newcommand{\Pc}{\operatorname{P}}

\newcommand{\hide}[1]{ }
\ifdefined\IsLLNCS
\else

\fi
\DeclareMathOperator*{\E}{\Exp}

\newcommand{\CC}{\mathrm{CC}}
\newcommand{\rv}[1]{\mathrm{#1}}
\DeclareMathOperator{\MI}{I}
\DeclareMathOperator{\HH}{H}

\newcommand{\nat}{\mathbb{N}}
\newcommand{\given}{\medspace | \medspace}
\newcommand{\Eve}{\mathrm{Eve}}

\newcommand{\Suc}[1]{\mathrm{AccGap}\left(#1\right)}
\newcommand{\Acc}[2]{\mathrm{Acc}_{#1}\left(#2\right)}
\newcommand{\SDC}[2]{\SD\left(\left(#1\right),\left(#2\right)\right)}
\newcommand{\SDL}[2]{\SD&(\left(#1\right),\\
	&\left(#2\right))}

\newcommand{\ISDCR}[2]{\SD\left(\left(#1\right)\times \left(#2\right),\left(#1,#2\right)\right)}

\newcommand{\SDI}{I_{SD}}

\newcommand{\ISD}[3]{\SDI\left(#1;#2|_{#3}\right)}
\newcommand{\ISDC}[2]{\SDI\left(#1;#2\right)}

\newcommand{\varcounter}{\mathsf{counter}}
\newcommand{\rvcounter}{\rv{\mathsf{Counter}}}
\newcommand{\varj}{j}

\newcommand{\x}{x}
\newcommand{\y}{y}
\newcommand{\Com}{\mathsf{Com}}
\newcommand{\Dist}{\mathsf{Dist}}
\newcommand{\Set}{\mathsf{Set}}

\newcommand{\Pitilde}{\Lambda_\Com}

\newcommand{\Atilde}{\Ac_\Com}
\newcommand{\Btilde}{\Bc_\Com}

\newcommand{\Pitag}{{\Pitilde'}}

\newcommand{\Atag}{{\Atilde'}}
\newcommand{\Btag}{{\Btilde'}}

\newcommand{\Piline}{{\Lambda_\Set}}
\newcommand{\Aline}{{\Ac_\Set}}
\newcommand{\Bline}{{\Bc_\Set}}

\newcommand{\Pz}{{\Lambda_z}}
\newcommand{\Az}{{\Ac_z}}
\newcommand{\Bz}{{\Bc_z}}

\newcommand{\Pcom}{\Lambda_\Com}
\newcommand{\Acom}{\Ac_\Com}
\newcommand{\Bcom}{\Bc_\Com}

\newcommand{\Pdist}{\Lambda_\Dist}
\newcommand{\Adist}{\Ac_\Dist}
\newcommand{\Bdist}{\Bc_\Dist}

\newcommand{\Pdisttag}{{\Pdist'}}
\newcommand{\Adisttag}{{\Adist'}}
\newcommand{\Bdisttag}{{\Bdist'}}

\newcommand{\setx}{{\cal{X}}}
\newcommand{\sety}{{\cal{Y}}}
\DeclareMathAlphabet \mathbfcal{OMS}{cmsy}{b}{n}

\newcommand{\setX}{{\rv{X}}}
\newcommand{\setY}{{\rv{Y}}}

\title{On the Communication Complexity of Key-Agreement  Protocols
	\Draft{\\\small \sc Working Draft: Please Do Not Distribute}\thanks{An extended abstract of this work appeared in ITCS 2019 \cite{haitner2019communication}.}
}

 \ifdefined\IsAnon
 \author{}
 \date{}
 \LLNCS{ \institute{}}
 \else
 \ifdefined\IsLLNCS
 \else
 \author{Iftach Haitner\thanks{The Blavatnik school of computer science, Tel Aviv University. E-mail: \texttt{iftachh@cs.tau.ac.il}. Member of  the Check Point Institute for Information Security.} %
 	\footnote{Supported by ERC starting grant 638121.}
 	\and Noam Mazor\footnotemark[3]~\thanks{The Blavatnik school of computer science, Tel Aviv University. E-mail: \texttt{noammaz@gmail.com}.}
 	\and Rotem Oshman \thanks{The Blavatnik school of computer science, Tel Aviv University. E-mail: \texttt{rotem.oshman@gmail.com}. Supported by the Israeli Centers of Research Excellence  program   4/11 and  BSF grant  2014256.} 
 	\and Omer Reingold \thanks{Computer Science Department, Stanford University. E-mail: \texttt{reingold@stanford.edu}. Supported by NSF grant CCF-1749750.}
 	\and Amir Yehudayoff \thanks{Department of Mathematics, Technion-Israel Institute of Technology. E-mail: \texttt{amir.yehudayoff@gmail.com}. Supported by ISF grant 1162/15.} 
 } 
 \fi
 \fi

\begin{document}

\begin{titlepage}

\maketitle

\begin{abstract}
	Key-agreement protocols whose security is proven in the \emph{random oracle model} are an important alternative
	to protocols based on public-key cryptography.
	In the random oracle model, the parties and the eavesdropper have access to a shared random function (an ``oracle''),
	but the parties are limited in the number of queries they can make to the oracle.
	The random oracle serves as an abstraction for black-box access to a symmetric cryptographic primitive, such as a collision resistant hash.
	Unfortunately, as shown by \citeauthor{ImpagliazzoRu89} [STOC '89]  and \citeauthor{BarakGhidary} [Crypto '09], such protocols can only guarantee  limited  secrecy:  the key of any  $\ell$-query protocol can be revealed by an   $O(\ell^2)$-query adversary. This quadratic gap between the query complexity of the honest parties and the eavesdropper matches the  gap obtained by the  \textit{Merkle's Puzzles} protocol of \citeauthor{Merkle82a}   [CACM '78].

	In this work we tackle a new aspect of key-agreement protocols in the random oracle model: their \emph{communication complexity}.
	In Merkle's Puzzles, to obtain secrecy against an eavesdropper that makes roughly $\ell^2$ queries,
	the honest parties need to exchange $\Omega(\ell)$ bits.
	We show that for protocols with certain natural properties, ones that  Merkle's Puzzle has,   such high communication is unavoidable. Specifically, this is the case  if the honest parties' queries are uniformly random, or alternatively if the protocol uses non-adaptive queries and has only two rounds.
	Our proof for the first setting uses a novel reduction from the set-disjointness problem in two-party communication complexity.
	For the second setting we prove the lower bound directly, using information-theoretic arguments.

Understanding the communication complexity of  protocols whose security is  proven in the random-oracle model is an important question in the study of practical protocols. Our results and proof techniques are a first step in this direction.
\end{abstract}


\thispagestyle{empty}
\end{titlepage}

\section{Introduction}\label{sec:Intro}
In a key-agreement protocol~\cite{DiffieHe76}, two parties  communicating over an insecure channel want to securely agree on a shared secret key,
such that an eavesdropper observing their communication cannot find the  key.
For example, given a hash function $h : [n] \rightarrow [N]$ that is hard to invert,
the players can execute the following protocol, called \emph{Merkle's puzzles}~\cite{Merkle82a}:
we fix an arbitrary parameter $\ell \approx \sqrt{n}$, and the parties select uniformly random subsets $A = \set{a_1,\ldots,a_{\ell}},B = \set{b_1,\ldots,b_{\ell}} \subseteq [n]$ (respectively) of size $\ell$.
We choose $\ell, n$ such that with constant probability there is a unique intersection, $|A \cap B| = 1$.
The first party evaluates $h$ on every element $a \in A$,
and sends $h(a_1),\ldots,h(a_{\ell})$ to the second party,
which then looks for a unique element $b \in B$ such that $h(b) = h(a_i)$ for some $i \in [\ell]$.
If found, the second party sends the index $i$ to the first party and outputs $b$ as the secret key;
the second party outputs $a_i$ as the secret key.
Because $h$ is a ``good'' hash function and $h(b) = h(a_i)$, it is likely that $b = a_i$, so the players output the same key.
Moreover, since $h$ is hard to invert, an eavesdropper that tries to find the secret key after seeing $h(a_1),\ldots,h(a_{\ell}), i$
must essentially compute $h$ on the entire universe in order to invert $h$ and find $a_i$.
Thus, we have a quadratic gap between the work performed by the eavesdropper, which must compute $\Omega(\ell^2)$ hashes, and the work performed by the parties, which compute $\ell$ hashes each.

Ideally we would strive for an \emph{exponential} gap between the work required to break the security of the protocol and the work of the honest parties.
There are numerous candidate constructions of such key-agreement schemes, \eg  \cite{RSA76,rabin1979digitalized,AjtaiDw97,mceliece1978public}, based on assumptions implying that \emph{public-key encryption schemes} exist.
A fundamental open question is whether we can design key-agreement protocols based on the security of symmetric primitives, \eg \emph{private-key} encryption (e.g., collision resistant hash);
the security of such primitives is believed to be more robust than public-key encryption.
A very important step in this direction was made by \citet{BarakGhidary} (following \citet{ImpagliazzoRu89}): they showed that as long as the symmetric primitive is used as a black box, the quadratic gap achieved by Merkle's puzzles is the best possible.

The notion of ``black box'' is formalized by the \emph{random oracle model}: instead of a concrete hash function $h$, we assume that the parties have access to a \emph{random oracle} $\rv{F} : [n] \rightarrow [n]$, a perfectly random function. The random oracle is ``the best hash function possible'' (w.h.p.), so lower bounds proven in the random oracle model hold for any instantiation where the oracle is replaced by a one-way function.
Thus, the lower bound of~\citet{BarakGhidary} rules out any black-box key-agreement scheme from one-way functions that achieves a better than quadratic gap between the eavesdropper's work and the honest parties.

While a quadratic gap  between the $\ell$-query honest parties and  the $\ell^2$-query eavesdropper might not seem like much,
and ideally we would wish for an exponential gap,
on modern architecture it can yield a good enough advantage, assuming that security is preserved when the random oracle is replaced with a fixed hash function.
For example, 
a consumer-level CPU (Intel Core i5-6600) can compute 5 million SHA-256 hashes per second, and specialized hardware for SHA-256 computation (for example, AntMiner S9) can compute $14\times 10^{12}$ hashes per second~\cite{bernstein2009ebacs}. 
It follows that  if the honest parties spend one  second of computation on standard CPU, an attacker with specialized hardware can violate the security of Merkle's puzzles in less than a second. However, if the parties spend one second on specialized hardware, an attacker with  specialized hardware has to spend more than $200,000$ years to break the scheme.

So, are Merkle's puzzles a practical and realistic key-agreement scheme? The answer is probably not: even setting aside the question of replacing the random oracle by a concrete hash function,  in Merkle's puzzles,
the honest parties \emph{send each other $\widetilde{\Omega}(\ell)$ bits} to obtain security against an eavesdropper that makes roughly $\ell^2$ queries.
In our example above, if we instantiate Merkle's puzzles using SHA-256 for one second on specialized hardware, the first party would need to send more than 100 terabytes to the second party.
A fundamental question is whether this high communication burden is inherent to secure key-agreement, and more generally, what is the communication cost of cryptographic protocols in the random oracle model and other oracle models.
In this paper we initiate the study of the communication complexity of cryptographic protocols in the random-oracle model.

\subsection{Our Results}
We show that for random-oracle protocols with certain natural properties, the high communication incurred by Merkle's puzzles is unavoidable: in order to achieve security against an adversary that can ask $\Theta(\ell^2)$ queries, the two parties must exchange $\Omega(\ell)$ bits of communication. Specifically, we show that the bound above holds for protocols where the parties' queries are a uniformly random set, and also for two-round  protocols that make non-adaptive (but arbitrary) queries.\footnote{These are both properties of Merkle's puzzles.}

To simplify the statements of our results, we focus here  on  key-agreement protocols whose \emph{agreement} parameter,
the probability that the players output the same key, is larger by some constant than their \emph{secrecy} parameter,
the probability that an eavesdropper can find the key.

 \paragraph{Uniform-query protocols.} 
 We say that a random-oracle protocol makes \emph{uniform queries} if each party's oracle queries are a uniformly random set.
 We give the following lower bound on the communication complexity of such protocols.

\begin{theorem}[lower bound on uniform-queries protocols, informal]\label{thm:mainUniformIntro}
	Any  $\ell$-uniform-query  key-agreement protocol  achieving  non-trivial secrecy against  $o(\ell^2)$-query  adversaries has  communication complexity   $\Omega(\ell)$. 
\end{theorem}
\cref{thm:mainUniformIntro} is proved by a reduction from \textit{set-disjointness},
a  problem in communication complexity that is known to require high communication.

\paragraph{Two-round non-adaptive protocols.}
An oracle protocol is said to make \emph{non-adaptive} queries
if the distribution of queries made by the players is fixed in advance, i.e.,
it is determined before the parties communicate with each other and does not depend on the oracle's answers.
We give the following lower bound on the communication complexity of such protocols.

\begin{theorem}[lower bound on  two-message non-adaptive protocols, informal]\label{thm:MainArbitraryIntro}
	Any two-message  $\ell$-query non-adaptive  key-agreement protocol  of non-trivial secrecy against $q$-query adversaries has  communication complexity   $\Omega(q/\ell)$. 
	\label{thm:lower2_informal}
\end{theorem}
Once again this lower bound is nearly-tight with Merkle's puzzles, where $q = \Theta(\ell^2)$, and the communication cost is $\tilde{\Theta}(\ell)$.\footnote{This theorem is also nearly-tight for any $q$, with a version of Merkle's puzzle, in which Alice is sending $\Theta(q/l)$ answers, from a universe of size $\Theta(q)$.}

Following~\citet{BarakGhidary} and \citet{ImpagliazzoRu89},
we prove this lower bound by presenting  an eavesdropper that makes $q$ queries and prevents the parties from exploiting the advantage they gain
by their joint random oracle calls.

In \cite{BarakGhidary}, the communication cost of the protocol is not taken into account: their eavesdropper makes $O(\ell^2)$ queries and has high probability of finding \emph{all} intersection queries (i.e., all queries that were asked by both players).
In our case, if the protocol has communication cost $C$, then to prove Theorem~\ref{thm:lower2_informal},
our eavesdropper must make only $O(C \cdot \ell)$ queries (to show the trade-off that $C = \Omega(q/\ell)$).
If $C \ll \ell$,
our eavesdropper makes much fewer queries than the eavesdropper in~\cite{BarakGhidary,ImpagliazzoRu89},
and in particular it cannot discover all the intersection queries.
Instead, our eavesdropper asks only queries that the players \emph{were able to learn} are in their intersection.
If a query is in the intersection, but the players have not communicated this fact to each other,
then the eavesdropper will not necessarily ask this query (unlike~\cite{BarakGhidary,ImpagliazzoRu89}).
Finding the correct definition for what it means to ``learn'' that a given query is in the intersection,
and constructing an eavesdropper that makes only $O(C \cdot \ell)$ queries,
are the main difficulty in our proof. \footnote{A lower bound on key-agreement protocols implies a lower bound for the Set-Intersection problem. This fact suggests that the proof cannot be simple.}

\subsection{Related Work}
\citet{ImpagliazzoRu89}  showed that the key of any  $\ell$-query key-agreement  protocol in the random-oracle model can be revealed by an $\tilde{O}(\ell^6)$ query eavesdropper. \citet{BarakGhidary} improve this bound and present an  $O(\ell^2)$ query eavesdropper for this task, which shows that Merkle puzzles  is optimal in this respect.
\citet*{HaitnerOZ12} used the machinery of \cite{BarakGhidary} to relate the security of protocols that do not use a random oracle and solve tasks with no input, to the security of no-input protocols in the random-oracle model against an $O(\ell^2)$-query adversary.
Finding limitation on the usefulness of random oracles for  protocols that do take input seems to be  a more difficult question. \citet{chor1991zero} and \citet{mahmoody2012limits} made some progress in this direction. Finally, \citet*{HaitnerHRS15} gave lower bounds on the  communication complexity of  statistically hiding commitments and single-server private information retrieval in a weaker oracle model that  captures the hardness of one-way functions/permutation more closely than the random-oracle model.

\subsection{Organization}
We begin by giving a high-level overview of our proof techniques in \cref{sec:intro:Technique:Uni,sec:intro:Technique:2msg}.
Formal definitions and notation used throughout the paper are given in \cref{sec:Preliminaries}. The bound for uniform-query  protocols is formally stated and proved in \cref{sec:UnifromProtocols}, and the bound for  two-message non-adaptive  protocols is stated and proved in \cref{sec:ArbitraryProtocols}. 
\ifdefined\IsLLNCS
Missing proofs are given in \cref{sec:Appendix}.
\fi


\section{Uniform-Query Protocols: Proof Outline} \label{sec:intro:Technique:Uni}
Our lower bound for uniform-query key-agreement  protocols is proved  via  a reduction to  \textit{set disjointness},
a classical problem in two-party communication complexity.

In the set disjointness problem, we have two players, Alice and Bob.
The players receive inputs $X,Y \subseteq [n]$, respectively, of size $|X| = |Y| = \ell$,
and
the players must determine whether $X \cap Y = \emptyset$.
To do this, the players communicate with each other,
and the question is how many bits they must exchange.
It is known~\cite{razborov1992distributional} that for any sufficiently large $n \in \N$,
if the size of the sets is $\ell = n/4$, then
the players must exchange $\Omega(n)$ bits to solve set disjointness,
and this holds even for randomized protocols where the players have access to shared randomness and only need to succeed with probability $2/3$.
Here, we require high success probability \emph{on any input}, not over some specific input distribution.
We note that in the 2-party communication complexity model there is no random oracle.

The connection between set disjointness and key agreement comes from the fact that the only \emph{correlation} between
the parties' views in a key agreement protocol comes from the \emph{intersection queries}, the queries that both players ask and
Eve does not know.
Indeed, if Alice asks $A \subseteq [n]$ and Bob asks $B \subseteq [n]$,
and the random oracle is $F : [n] \rightarrow [n]$,
then $F(A \setminus (A \cap B))$ and $F(B \setminus (A \cap B))$ are \emph{independent of each other}.
In particular, if $A \cap B = \emptyset$, then $F(A)$ and $F(B)$ are independent, and intuitively, in this case the players cannot securely agree on a secret key, because they have no advantage over the eavesdropper.
On the other hand, if $A \cap B \neq \emptyset$, then the players can exploit the correlation induced by $F(A \cap B)$ to securely agree on a secret key.
Thus, any secure key agreement protocol ``behaves differently'' depending on whether $A \cap B = \emptyset$ or not, and
we can use this to solve the set disjointness problem.

Suppose that we are given a secure key-agreement protocol $\Pi$, where the players make $\ell$ uniformly-random queries to an oracle $\rv{F} : [n] \rightarrow [n]$.
For simplicity we assume that the protocol has \emph{perfect agreement}, that is, the players always output the same key,
and that the security parameter is $3/4$, that is, an eavesdropper has probability at most $3/4$ of outputting the same key as the players.
Our full proof does not make these assumptions.

Now, we want to construct from the key-agreement protocol $\Pi$, which uses a random oracle, a protocol $\Pi'$ for set disjointness, \emph{without} a random oracle (as usual in communication complexity).
To this end, we consider two possible ways of simulating $\Pi$ without an oracle:
\begin{itemize}
	\item $\Pcom$: the players use their shared randomness to simulate the oracle.
		They interpret the shared randomness as a random function $\rv{F} : [n] \rightarrow [n]$,
		and whenever $\Pi$ wants to query some element $q \in [n]$,
		the players use $\rv{F}(q)$ as the oracle's answer.
	\item $\Pdist$: the players use their \emph{private} randomness to simulate the oracle.
		Alice and Bob interpret their private randomness as random functions $\rv{F}_A, \rv{F}_B : [n] \rightarrow [n]$,
		respectively.
		Whenever $\Pi$ indicates that Alice should query an element $q \in [n]$, she uses $\rv{F}_A(q)$ as the answer,
		while Bob uses $\rv{F}_B(q)$.
\end{itemize}
The first simulation, $\Pcom$, is ``perfect'': it produces exactly the correct distribution of transcripts and outputs under our key-agreement protocol $\Pi$.
In particular, the keys produced by the players in $\Pcom$ always agree, and an eavesdropper that sees the transcript of $\Pcom$ (but not the shared randomness) can find the key with probability at most $3/4$.

On the other hand, the second simulation $\Pdist$ is ``wrong'', because the players do not use the same random function to simulate the random oracle.
In fact, it is known that without shared randomness, secure key agreement is \emph{impossible},
as an eavesdropper that sees the transcript can find the key with the same probability that the players have of agreeing with each other.
Therefore there are two possible cases:
\begin{description}
	\item[Agreement gap:] The probability that the players agree on the key in $\Pdist$ is at most $7/8$ (compared to one  in $\Pcom$), or
	
	\item[Secrecy gap:] There is an eavesdropper $\Ec$ that guesses Alice's key in $\Pdist$ with probability at least $7/8$ (compared to $3/4$  in $\Pcom$).
	
\end{description}
(Instead of $7/8$ we could have used here any constant probability in $(3/4, 1)$, but in the full proof this choice depends on the agreement and security parameters of $\Pi$.)

We divide into cases, depending on which of the two gaps we have.

\paragraph{Agreement gap.}
Assume that the players agree with probability at most $7/8$ in $\Pdist$. 
For simplicity, let us make the stronger assumption that for any intersection size $c > 0$,
the probability of agreement between the players is at most $7/8$, even \emph{conditioned} on the event that $|A \cap B| = c$.
A general key-agreement protocol might not satisfy this assumption, which complicates the full proof significantly;
see Section~\cref{sec:UnifromProtocols} for the details.

So, we assumed that whenever the intersection is non-empty, the players agree with probability at most $7/8$.
Observe, however, that when the intersection \emph{is} empty ($A \cap B = \emptyset$),
the distribution of transcript and outputs in $\Pdist$ is the same as in $\Pi$: although each player uses a different random function, they never ask the same query, so there is no inconsistency.
Therefore, conditioned on $A \cap B = \emptyset$,
in $\Pdist$ the players have perfect agreement (as in $\Pi$).
In other words, $\Pdist$ behaves very differently when $A \cap B = \emptyset$, in which case the players always agree on the key,
compared to the general case, where the players agree with probability at most $7/8$. We use this fact to \emph{check} whether $A \cap B = \emptyset$.
Thus, by checking whether or not they got the same key in $\Pdist$,
the players get an indication for whether or not $A \cap B = \emptyset$.

Our set disjointness protocol $\Pi'$ is defined as follows.
Given inputs $X,Y \subseteq [n]$, respectively,
the players simulate $\Pdist$ several times.
In each simulation, the players agree on a random permutation $\sigma : [n] \rightarrow [n]$ using their shared randomness,
and then the players simulate $\Pdist$ using their permuted inputs as the query set;
that is, Alice feeds $A = \sigma(X)$ to $\Pdist$ as her query set, and Bob feeds $B = \sigma(Y)$ to $\Pdist$ as his query set.
Note that $A,B$ are uniformly random, \emph{subject to} having an intersection of size $|X \cap Y|$.

After each simulation of $\Pdist$, the players send each other the keys output under $\Pdist$,
and check if they got the same key.
Finally, they output ``$X \cap Y = \emptyset$'' iff they got the same key in all the simulations of $\Pdist$.

Since $\Pdist$ has perfect agreement when there is no intersection,
the players always succeed when $X \cap Y = \emptyset$.
However, by assumption, whenever $X \cap Y \neq \emptyset$,
the probability of agreement in $\Pdist$ is at most $7/8$, so if we repeat $\Pdist$ sufficiently many times,
the probability that all instances output the same key will be at most $1/3$.

\paragraph{Secrecy gap.}
In this case we convert  $\Pcom$ and  $\Pdist$ into a pair of protocols with an agreement gap, and then proceed as above.

Consider the protocol $\Pdisttag$ where the parties acts as in $\Pdist$,  but at the end,
Bob executes the eavesdropper $\Ec$ on the transcript, and outputs the key that $\Ec$ outputs.
Define $\Pcom'$ analogously.

By assumption, $\Ec$ guesses Alice's output in $\Pdist$ with probability at least $7/8$, but in $\Pcom$ it succeeds with probability at most $3/4$.
Thus, in $\Pdist'$ the players agree with probability at least $7/8$, but in $\Pcom'$ they agree with probability at most $3/4$;
there is a gap of at least $1/8$ between the probability of agreement in the two protocols (although they have switched roles and now $\Pdist'$ has the higher agreement probability).
Note also that $\Pdist'$ does not have agreement probability 1, as we assumed for simplicity above, but our full proof can handle this case.

\paragraph{What about general protocols?}
It was important for our reduction to assume that the key-agreement protocol makes uniformly-random queries.
Indeed, this reduction fails in the general case:
consider the protocol where Alice and Bob always query 1, and output $\rv{F}(1)$ as their secret key.
This protocol is completely insecure, since the eavesdropper can also query 1 and output $\rv{F}(1)$.
But our reduction would not work for it, because the input distribution where both players get the set $\set{1}$ \emph{is not hard for set disjointness}
(indeed it is trivial).
We see that the ``hardness'' of secure key-agreement is not necessarily that it is hard for the players to find their intersection queries,
but that the eavesdropper should not be able to \emph{predict} the intersection queries that the players use.
Our second lower bound makes this intuition explicit and uses it to get a lower bound on two-round protocols with arbitrary (but non-adaptive) query distributions.

\section{Two-Message Non-Adaptive Protocols: Proof Outline}\label{sec:intro:Technique:2msg}

In this section we describe a lower bound on the communication cost of any key-agreement protocol that makes non-adaptive queries and uses two rounds of communication: we show that any such protocol that makes $\ell$ queries and is secure against an adversary that makes $q$ queries must send a total of $\Omega(q/\ell)$ bits. In particular, taking $q = \Theta(\ell^2)$, this shows that Merkle's puzzles is optimal in its communication cost.

In this proof, we once again relate the parties' advantage over the eavesdropper to the information they gained about the intersection of their query sets.
We show that  to produce a shared key, the parties need to learn a lot of information about this intersection. Moreover, the query sets and their intersection need to be ``unpredictable'' (have high min-entropy) given the transcript, otherwise an eavesdropper could make the same queries and output the same key. 

\paragraph{Preliminaries.}
In the proof we often need to measure differences between various distributions.
For this purpose we use \emph{$f$-divergences}:
given a convex function $f: \mathbb{R} \rightarrow \mathbb{R}$ with $f(1) = 0$,
and distributions $P, Q$,
the \emph{$f$-divergence of $P$ from $Q$} is defined as
\begin{equation*}
	\mathsf{D}_f(\rv{P} \parallel \rv{Q})=\sum_{q\in \rv{Q}}\pr{\rv{Q}=q}f\left(\frac{\pr{\rv{P}=q}}{\pr{\rv{Q}=q}}\right).
\end{equation*}
Specifically, the two $f$-divergences we use in this paper are the \emph{statistical distance}, obtained by taking $f(x) = |x-1|/2$,
and the \emph{KL divergence}, obtained by taking $f(x) = x \log x$.
Each has its own nice properties and disadvantages:
statistical distance is bounded in $[0,1]$ but it is not additive,
while KL divergence is additive but unbounded (more on this below).

We frequently need to measure the ``amount of dependence'' between two random variables.
Let $(\rv{X}, \rv{Y}) \sim P_{\rv{X}, \rv{Y}}$ be random variables jointly distributed according to $P_{\rv{X}, \rv{Y}}$, and let $P_{\rv{X}}, P_{\rv{Y}}$ be the marginal distribution of $\rv{X}$ and $\rv{Y}$, respectively.
Also, let $P_{\rv{X}} \times P_{\rv{Y}}$ be the product distribution where $\rv{X}$ and $\rv{Y}$ are sampled independently of each other, each from its marginal distribution $P_{\rv{X}}, P_{\rv{Y}}$ (respectively).
To quantify the dependence between $\rv{X}$ and $\rv{Y}$, we measure the difference between their joint distribution and the product of the marginals:
formally, we define
\begin{equation*}
	\MI_f( \rv{X} ; \rv{Y} ) = \mathsf{D}_f( P_{\rv{X},\rv{Y}} \parallel P_{\rv{X}} \times P_{\rv{Y}} ).
\end{equation*}
This generalizes the usual notion of mutual information, which is the special case of $\MI_f$ where
we use KL divergence (i.e., when $f = x \log x$).
For clarity, when we use KL divergence we omit the subscript $f$,
and when using statistical distance, we use the notation $\MI_{SD}$ (instead of $\MI_{f(x) = |x-1|/2}$).

Finally, we also need the notion of \emph{conditional mutual information},
which is simply the average mutual information between two variables $\rv{X}, \rv{Y}$, where the average
is taken over a third random variable $\rv{Z}$.
Formally, let $(\rv{X},\rv{Y},\rv{Z}) \sim P_{\rv{X},\rv{Y},\rv{Z}}$.
For any value $z$,
let $P_{\rv{X},\rv{Y}|\rv{Z} = z}, P_{\rv{X} | \rv{Z} = z}, P_{\rv{Y} | \rv{Z} = z}$
be the joint distribution of $\rv{X},\rv{Y}$ and the marginals of $\rv{X}$ and $\rv{Y}$, respectively,
all conditioned on the event $\rv{Z} = z$.
Then we define
$\MI_f( \rv{X} ; \rv{Y} | \rv{Z} ) = \E_{z \sim P_{\rv{Z}}} \left[ \mathsf{D}_f( P_{\rv{X},\rv{Y}|\rv{Z} = z} \parallel P_{\rv{X}|\rv{Z}=z} \times P_{\rv{Y}|\rv{Z}=z} )\right]$.

\paragraph{Some examples.}

Let us illustrate the ideas behind the lower bound by way of some examples.

\paragraph{Example 1:}
We already discussed the na\:ive example where both players query $1$ and output $\rv{F}(1)$,
and said that it is insecure because the eavesdropper can \emph{predict} the intersection query.
Here is another instantiation of this idea:
Alice and Bob view the domain $\ell^2$ as an $\ell \times \ell$ matrix,
so that the oracle queries are represented by pairs $(i,j) \in [\ell]^2$.
Alice chooses a row $a \in [\ell]$, and queries all the elements of the row (that is, all pairs $(a, j)$ where $j \in [\ell]$);
Bob chooses a column $b \in [\ell]$ and queries all the elements of the column (all pairs $(i,b)$ where $i \in [\ell]$).
Then, Alice sends $a$ to Bob, who responds with $\rv{F}(a,b)$.
From $\rv{F}(a,b)$, Alice can compute $b$, by finding the (w.h.p.\ unique) index $j$ such that $\rv{F}(a,b) = \rv{F}(a,j)$.
Both players output the first bit of $b$ as the key.

This protocol is slightly less na\"ive than the previous one: now there are no queries that have high prior probability of being asked,
and the index $b$ of the query that determines the key is uniformly random a-priori.
However, once Alice sends $a$ to Bob, the game is up:
Eve can also query row $a$ and find $b$ the same way Alice does.

We see that in addition to queries that have a high prior probability of being asked,
Eve also needs to ask queries that have a high \emph{posterior} probability of being asked,
after she sees $\rv{M}_1$.
It turns out that this is enough: if we were to continue for more than 2 rounds, then Eve would also need to ask queries that become likely after seeing
$\rv{M}_2$, and so on, but to prove a 2-round lower bound, Eve does not need to ask these queries.
Intuitively, if a query only becomes likely after $\rv{M}_2$ is sent, then this is ``too late'' for it to be useful to the players,
and Eve can ignore it.

\paragraph{Example 2:}
First, both players query 1. Then they carry out the protocol from Example 1, but all messages are ``encrypted''
by XOR-ing them with $\rv{F}(1)$.

From this example we see that Eve needs to be somewhat adaptive: when she decides what queries to ask after seeing $\rv{M}_1$,
she must incorporate the queries she asked before the first round (in this case, she would query 1).
Essentially, when Eve tries to understand what the players have done in round $i$, she should take into account all the queries she made up to round $i$.

Should Eve be adaptive \emph{inside} each round?
In other words, after seeing $\rv{M}_1$, should she ask all queries $\cE_1$ that became likely, then compute which new queries are now likely given
$\rv{M}_1, \cE_1$, and so on, until she reaches a fixpoint?

It turns out that for our purposes here, because we consider non-adaptive protocols, Eve does not need to do this.


\paragraph{Heavy queries.}
Our attacker Eve tries to break the security of the protocol by asking all queries that are ``somewhat likely'' to be asked by the players;
these queries are called \emph{heavy queries}.
Informally, a query $q \in \zn$ is \emph{heavy} after round $i$ if given the transcript up to round $i$ (inclusive),
and given Eve's queries up to round $i$,
the probability that $q$ is asked by one (or both) of the players exceeds some threshold $\delta$ which is fixed in advance.

More formally, the set $\cE_i$ of heavy queries after round $i$ is defined by induction on rounds, as follows:
the a-priory heavy queries, $\cE_0$, are given by
\begin{equation*}
				\cE_{0} =\set{q \in \zn \colon  \pr{q \in X \cup Y  } \geq \delta}.
\end{equation*}
These are queries that are ``somewhat likely'' to be asked before the protocol begins.
For $i > 0$,
we define
\begin{equation*}
	\cE_{i} = \cE_{i-1} \cup \set{q \in \zn \colon  \pr{q \in X \cup Y \given \medspace M_{\leq i}, \rv{F}(\cE_{i-1})} \geq \delta}.
\end{equation*}
In other words, after round $i$, Eve asks all queries $q \in \zn$ that have probability at least $\delta$ of being queried by the players,
given the messages $M_{\leq i}$ that Eve observed up to round $i$ and the heavy queries she asked before, $\cE_{i-1}$.

\paragraph{A simplified normal form for protocols.}
To simplify the proof of the lower bound, we first apply an easy transformation to the protocol:
given a key agreement protocol $\Pi$, we construct a protocol $\Pi'$, which has nearly the same communication and query complexity as $\Pi$,
the same number of rounds,
and the same agreement and nearly the same security parameters.
But $\Pi'$ also has the following properties:
first, $\Pi'$ has no a-priori heavy queries, that is, $\cE_0 = \emptyset$;
and second, the secret key output by Bob in $\Pi'$ is the first bit of Bob's last query.
This easy transformation is omitted in this overview.

\subsection{Measuring the Players' Advantage Over Eve}

As we saw in the examples above, the players' ability to produce a shared secret key is closely tied to how much information
\emph{the players} have that \emph{Eve does not have} about the intersection of the query sets, $X \cap Y$.

To quantify this advantage,
define the following random variables:
\begin{itemize}
	\item $\rv{S}_i = \left( \rv{X} \cap \rv{Y} \right) \setminus \cE_{i-1}$, the intersection queries that have not been asked by Eve.
	\item $\rv{F}(\rv{S}_i)$: the answers to the queries in $\rv{S}_i$.
	\item $V_E^i = (\cE_i', \rv{F}(\cE_i'))$: a subset of the heavy queries for the previous round, and the answers to them.
Here, $\cE_i' \subseteq \cE_i$ is a subset that will be defined later (and depends on the round number $i$).
For technical reasons,
it is convenient to use only some of the heavy queries in some contexts; (essentially, in some places in our proof, Eve uses only some of her power.
	This helps us avoid some unnecessary dependencies.
\end{itemize}

We measure the advantage gained by the players in round $i$ by an expression of the form:
\begin{equation}
	\MI_f( \rv{S}_i, \rv{F}(\rv{S}_i) ; \rv{M}_i | \rv{Z}, \rv{V}_E^{i-1}, \rv{M}_{<i} ),
	\label{eq:advantage}
\end{equation}
where $\MI_f$ is the information with respect to the $f$-divergence,
$\rv{M}_i, \rv{M}_{<i}$ are the $i$-th message and the messages of rounds $1,\ldots,i-1$,
$\rv{Z}$ is the query set of the player that sent $\rv{M}_i$ (either $\rv{X}$ or $\rv{Y}$, depending on the round number $i$).
Note that Eve uses her heavy queries from the previous rounds, $\rv{V}_E^{i-1}$, to ``try to understand'' what is going on in the current round.

Intuitively, this expression measures how much information the $i$-th message conveys about the intersection queries and their answers,
which Eve \emph{cannot guess}.
For this reason, the random variable $\rv{S}_i$ excludes intersection queries that were asked by Eve.
Notice that on the right-hand side we condition on Eve's view (or on things Eve can sample):
Eve has already seen the messages $\rv{M}_{<i}$ and asked the heavy queries $\rv{V}_E^{i-1} = (\cE_{i-1}', \rv{F}(\cE_{i-1}'))$,
and 
she can sample the queries $\rv{Z}$, either $\rv{X}$ or $\rv{Y}$,
from the correct distribution given the transcript and her queries.
Crucially, this does not require her to make any oracle queries: we do not require her to sample the answers $\rv{F}(\rv{Z})$, only the queries $\rv{Z}$.
In other words, Eve can \emph{pretend} to be whichever player the query set $\rv{Z}$
belongs to, and by conditioning on her view, we essentially neutralize all the information that Eve can extract about the intersection.
Thus, the expression in~\eqref{eq:advantage}
measures the information the players gain about the intersection but that is hidden from Eve.%
\footnote{As we said above, we use only some of Eve's heavy queries, $\cE_{i-1}' \subseteq \cE_i$, so this intuition is not completely accurate;
	specifically, when~\eqref{eq:advantage} is \emph{large},
	it does not mean that Eve cannot guess a lot about the intersection, because she could use the full set $\cE_i$.
	However, when~\eqref{eq:advantage} is \emph{small},
then indeed Eve knows almost as much about the intersection as the players do, because her view \emph{includes} $\rv{V}_E^{i-1}$ (and possibly more).}

Our proof consists of showing:
\begin{enumerate}[Step I:]
	\item After the first message $\rv{M}_1$ is sent, the advantage gained is small, only $O(\delta |\rv{M}_1|)$.
		For this part of the proof we use KL-divergence to measure the advantage.
	\item After the second message $\rv{M}_2$ is sent, the advantage is still small, only \\$O(\sqrt{ \delta(|\rv{M}_1| + |\rv{M}_2|)})$.
		Here we use statistical distance to measure the advantage, for reasons we will explain below.
	\item When the expression in~\eqref{eq:advantage} is small (i.e., the players only have a small ``advantage''), then indeed,
		Eve can break the security of the protocol, by pretending to be one of the players and sampling the secret key that this player would output.
\end{enumerate}

Next we explain in more detail how each step is carried out.

\subsection{Outline of the Proof}

\paragraph{Step III: How Eve breaks security.}
Let us start from the end: suppose that after the second round, the ``advantage'' is small:
\begin{equation*}
	\MI_f( \rv{S}_2, \rv{F}(\rv{S}_2) ; \rv{M}_2 | \rv{Y}, \rv{M}_1, \rv{V}_E^1) \leq \beta,
\end{equation*}
where $\beta = O(\sqrt{ \delta (|\rv{M}_1| + |\rv{M}_2|) }) \ll 1$.
Here, the advantage is measured in statistical distance (that is, we take $f(t) = |t - 1|/2$).
We want to show that Eve can break the security of the protocol, by guessing the secret key.

As we said, Eve's strategy is to ``pretend'' that she is Bob,
and sample Bob's output, $\out^{\Bc}$.

In a general protocol, to do this, Eve needs to sample Bob's queries $\rv{Y}$ and the answers $\rv{F}(\rv{Y})$, 
and then she can compute $\out^{\Bc} = \out^{\Bc}(\rv{Y}, \rv{F}(\rv{Y}), \rv{M}_1, \rv{M}_2)$.
However, recall that we transformed the protocol so that $\out^{\Bc}$ is a fixed function of $\rv{Y}$;
therefore, Eve in fact needs to do nothing clever, only sample $\rv{Y}$ given her view $\rv{M}_1, \rv{M}_2, \cE_1, \rv{F}(\cE_1)$
and compute $\out^{\Bc}$ from $\rv{Y}$.

We need to show that Eve's key is close to the correct distribution, the one used by the players.
In general, if too much communication is allowed, this is not true,
as shown by the following example.
\fbox{\begin{minipage}{\textwidth}
		\textbf{Example:}
In Merkle's puzzles,
Alice's message is $\rv{F}(\rv{X})$,
and Bob responds with $\rv{F}(\rv{s})$,
where $\rv{s} \in \rv{X} \cap \rv{Y}$ is some intersection query.
The original secret key (before our transformation) is the first bit $\rv{s}^1$.
After our transformation,
the secret key is $\rv{Y}_{\ell+1}^1$,
and as part of $\rv{M}_2$, Bob sends Alice the bit $\rv{b} = \rv{s}^1 \xor \rv{Y}_{\ell+1}^1$ so that she can extract $\rv{Y}_{\ell+1}^1$.

From Alice's perspective, given $\rv{X}, \rv{F}(\rv{X})$,
Bob's message $\rv{M}_2 = \rv{F}(\rv{s}), \rv{b}$
fixes $\rv{Y}_{\ell+1}^1$ to the value $\rv{b} \xor \rv{s}^1$.
(We ignore here the tiny probability that $\rv{s}$ cannot be uniquely computed from $\rv{X}, \rv{F}(\rv{X})$ and $\rv{F}(\rv{s})$,
i.e., the probability of a collision in $\rv{F}$.)
However, from Eve's perspective, because she does not know $\rv{X},\rv{F}(\rv{X})$ and she asks no queries (there are no heavy queries in Merkle's puzzles),
the intersection element $\rv{s}$ remains uniformly random.
When Eve samples $\rv{Y}_{\ell+1}^1$ given $\rv{M}_1, \rv{M}_2$ and her non-existent heavy queries,
the result is random,
and completely independent from the true secret key.
\end{minipage}}

We need to show that when the players' advantage is small,
then the example above cannot happen, and Eve's key agrees with the players' w.h.p.
To this end, we are interested in the difference between Eve's ``pretend distribution'', and the true distribution that the players use to produce the key:
if the two distributions are close, then Eve's chances of guessing the right secret key are roughly the same as Bob's.
The \emph{only difference} between these two distributions is that 
given $\rv{M}_1, \rv{M}_2$ and $\cE_1, \rv{F}(\cE_1)$ (which the players do not use),
\begin{itemize}
	\item 
The players' keys are produced according to the \emph{joint distribution} $(\out^{\Ac}, \out^{\Bc})$,
and in particular, both players have the same answers $\rv{F}(\rv{S}_2)$ to the non-heavy intersection queries $\rv{S}_2 = \left(\rv{X} \cap \rv{Y}\right) \setminus \cE_1$.
\item Eve's pretense that she is Bob is carried out \emph{independently} from Alice's view:
Eve cannot use the true intersection queries (which she does not know), only what she has learned about them from
$\rv{M}_1, \rv{M}_2, \rv{V}_E^i$.
The joint distribution of Alice and Eve's keys is therefore given by the product distribution $\out^{\Ac} \times \out^{\Bc}$.
\end{itemize}
So, we would like to bound the difference between the joint distribution and the product distribution, i.e.,
\begin{equation*}
	\MI_f( \out^{\Ac}; \out^{\Bc} | \rv{M}_1, \rv{M}_2, \cE_1, \rv{F}(\cE_1)).
\end{equation*}

Given the conditioning, 
Alice's output $\out^{\Ac}$ is a function of her view, $\rv{X}, \rv{F}(\rv{X})$.
Also, we assumed that Bob's output is a function of his queries $\rv{Y}$.
Therefore,
\begin{equation}
	\MI_f( \out^{\Ac}; \out^{\Bc} | \rv{M}_1, \rv{M}_2, \cE_1, \rv{F}(\cE_1)) \leq \MI_f( \rv{X}, \rv{F}(\rv{X}) ; \rv{Y} | \rv{M}_1, \rv{M}_2,  \cE_1, \rv{F}(\cE_1)).
	\label{eq:diff}
\end{equation}
Now we need to show that 
given Eve's view, the dependence between $\rv{X}, \rv{F}(\rv{X})$ and $\rv{Y}$
is bounded in terms of the advantage:
\begin{equation}
	\MI_f( \rv{X}, \rv{F}(\rv{X}) ; \rv{Y} | \rv{M}_1, \rv{M}_2, \cE_1, \rv{F}(\cE_1))
	\leq 
	\MI_f( \rv{S}_2, \rv{F}(\rv{S}_2) ; \rv{M}_2 | \rv{M}_1,\rv{Y}, \cE_1 \cap \rv{Y}, \rv{F}(\cE_1 \cap \rv{Y})).
	\label{eq:adv_bound}
\end{equation}

This proof is somewhat tedious; it relies on the fact that $\rv{M}_2$ is a function of $\rv{M}_1, \rv{Y}$ and $\rv{F}(\rv{Y})$,
and on the fact that $\rv{X}, \rv{F}(\rv{X})$ are independent of $\rv{Y}, \rv{F}(\rv{Y})$ given the intersection queries and answers,
$\rv{S}_2, \rv{F}(\rv{S}_2)$ and $\cE_1, \rv{F}(\cE_1)$.
Intuitively, all the dependence between $\rv{X}, \rv{F}(\rv{X})$ and $\rv{Y}$ ``flows through'' what the players learn about the intersection,
and the proof of~\eqref{eq:adv_bound} formalizes this intuition.

\paragraph{Step I: Bounding the advantage after the first round.}
For the first round, we analyze the players' advantage in terms of KL-divergence, and bound
\begin{equation*}
	\MI( \rv{S}_1, \rv{F}(\rv{S}_1) ; \rv{M}_1 | \rv{X} ).
	\label{eq:adv0}
\end{equation*}
Notice that we do not use Eve at this point, because we eliminated any a-priory heavy queries, so there is nothing Eve needs to query
in order to ``understand'' $\rv{M}_1$.
For the same reason, $\rv{S}_1 = \rv{X} \cap \rv{Y}$ (there are no heavy queries to remove from the intersection).

We claim that
\begin{equation}
	\MI( \rv{S}_1, \rv{F}(\rv{S}_1) ; \rv{M}_1 | \rv{X} )
	\leq
	\delta |\rv{M}_1|.
	\label{eq:adv1}
\end{equation}
This is not hard to see: suppose $\rv{X} = x$.
Because we got rid of the a-priori heavy queries,
every individual query $q \in x$
has probability at most $\delta$ of being asked by Bob (otherwise, $q$ would be heavy).
Therefore, for every $q \in x$,
we have $\Pr\left[ q \in \rv{S}_1 | \rv{X} = x \right] \leq \delta$.
Because $\rv{M}_1$ is generated by Alice without knowing $\rv{S}_1$,
and every query is in $\rv{S}_1$ only w.p.\ at most $\delta$,
intuitively, the information in $\rv{M}_1$ ``only applies'' to the queries in $\rv{S}_1$ with probability $\delta$.
Therefore the information that $\rv{M}_1$ gives about $\rv{S}, \rv{F}(\rv{S}_1)$ is at most $\delta|\rv{M}_1|$.

The actual proof involves a Shearer-like argument for mutual information, similar to the ones used in \cite{ganor2015exponential, rao2015simplified}.

\paragraph{Step II: Bounding the advantage after the second round.}

Now we must bound the advantage the players gain after the second round,
and show that
\begin{equation}
	\MI_{SD}( \rv{S}_2, \rv{F}(\rv{S}_2) ; \rv{M}_2 | \rv{Y}, \rv{M}_1, \cE_1 \cap \rv{Y}, \rv{F}(\cE_1 \cap \rv{Y}))
	=
	O( \sqrt{ |\rv{M}_1| + |\rv{M}_2|} ).
	\label{eq:adv2}
\end{equation}
As we said, we switch here to using statistical distance, and we will see why below.

Following the first round, we know that not much is known about the intersection, because Alice's message $\rv{M}_1$
did not convey a lot of information about it.
So, our proof here proceeds in two steps:
first, we ``pretend'' that \emph{nothing} is known about the intersection,
and consider the distribution $\mu'$ where given $\rv{M}_1$ the distribution of $\rv{Y}, \rv{F}(\rv{Y})$ is completely independent from $\rv{X}$.
We show that under $\mu'$,
Bob's message $\rv{M}_2$ would only convey $\delta|\rv{M}_2|$ bits of information about the intersection.
This is very similar to the analysis of the first round, and it is also carried out using KL-divergence.
Formally, we show that for the distribution $\mu'$ where $\rv{Y}, \rv{F}(\rv{Y})$ are drawn independently of of $\rv{X}$,
we have
\begin{equation}
	\MI^{\mu'}( \rv{S}_2, \rv{F}(\rv{S}_2) ; \rv{M}_2 | \rv{Y}, \rv{M}_1, \cE_1 \cap \rv{Y}, \rv{F}(\cE_1 \cap \rv{Y})) \leq \delta|\rv{M}_2|.
	\label{eq:info1}
\end{equation}
The proof relies on the fact that we excluded heavy queries from $\rv{S}_2$
(recall that $\rv{S}_2 = \left( \rv{X} \cap \rv{Y} \right) \setminus \cE_1$),
so given the conditioning, any query in $\rv{Y}$ can only belong to $\rv{S}_2$ with probability at most $\delta$.

However, $\mu'$ is not the real distribution: given $\rv{M}_1$, we do know a little about the intersection,
so $\rv{Y}, \rv{F}(\rv{Y})$ are not completely independent from $\rv{X}$.
Our next step is to switch to statistical distance, and show that the real distribution $\mu$ (where $\rv{X}, \rv{Y}$ are not independent)
and $\mu'$ (where they are) are close to each other.
Therefore, what we showed for $\mu'$ is also true for $\mu$, with the addition of a small penalty corresponding to the distance between $\mu$ and $\mu'$.

Formally, we prove that 		
\begin{align}
\begin{aligned}
&\MI_{SD}^{\mu}( \rv{S}_2, \rv{F}(\rv{S}_2) ; \rv{M}_2 | \rv{Y}, \rv{M}_1, \cE_1 \cap \rv{Y}, \rv{F}(\cE_1 \cap \rv{Y}))\\
&\leq
O\left(
\MI_{SD}^{\mu'}( \rv{S}_2, \rv{F}(\rv{S}_2) ; \rv{M}_2 | \rv{Y}, \rv{M}_1, \cE_1 \cap \rv{Y}, \rv{F}(\cE_1 \cap \rv{Y}))
+ \mathsf{D}_{SD}( \mu' \parallel \mu)
\right)\end{aligned}\label{eq:SDprop}
\end{align}


Under $\mu'$, by~\eqref{eq:info1} and Pinsker's inequality, we have:
\begin{equation}
	\MI^{\mu'}_{SD}( \rv{S}_2, \rv{F}(\rv{S}_2) ; \rv{M}_2 | \rv{Y}, \rv{M}_1, \cE_1 \cap \rv{Y}, \rv{F}(\cE_1 \cap \rv{Y})) \leq \sqrt{\delta|\rv{M}_2|}.
	\label{eq:info2}
\end{equation}
So, under $\mu'$ the expected amount of information revealed is small.

Next, we bound the difference between $\mu$ and $\mu'$.
We show that:
	\begin{equation*}
		\MI( \rv{Y}, \rv{F}(\rv{Y}) ; \rv{X} | \rv{M}_1 ) \leq 
		\MI( \rv{S}_1, \rv{F}(\rv{S}_1) ; \rv{M}_1 | \rv{X} ).
	\end{equation*}
	This is quite similar to the proof of Step III above --- here we do use standard mutual information,
	so the proof uses the chain rule, just as we did above.
	Since we have shown in Step I that 
	$\MI( \rv{S}_1, \rv{F}(\rv{S}_1) ; \rv{M}_1 | \rv{X} ) \leq \delta|\rv{M}_1|$,
	we conclude using Pinsker's inequality that
\begin{equation}
	\mathsf{D}_{SD}( \mu' \parallel \mu) \leq \sqrt{ \mathsf{D}_{KL}( \mu' \parallel \mu) } \leq \sqrt{ \delta |\rv{M}_1| }.
	\label{eq:mu_mu'}
\end{equation}

Together,~\eqref{eq:info2} and~\eqref{eq:mu_mu'} are the ingredients we need to apply~\eqref{eq:SDprop},
and obtain:
		\begin{align*}
			&\MI_{SD}^{\mu}( \rv{S}_2, \rv{F}(\rv{S}_2) ; \rv{M}_2 | \rv{Y}, \rv{M}_1, \cE_1 \cap \rv{Y}, \rv{F}(\cE_1 \cap \rv{Y}))
			\\
			&\leq
			O\left(
			\MI_{SD}^{\mu'}( \rv{S}_2, \rv{F}(\rv{S}_2) ; \rv{M}_2 | \rv{Y}, \rv{M}_1, \cE_1 \cap \rv{Y}, \rv{F}(\cE_1 \cap \rv{Y}))
			+ \mathsf{D}_{SD}( \mu' \parallel \mu)
			\right)
			\\
			&\leq
			O(\sqrt{\delta(|\rv{M}_1|+|\rv{M}_2|)}).
		\end{align*}

\section{Preliminaries}\label{sec:Preliminaries}

\subsection{Notations}
We use calligraphic letters to denote sets, uppercase for random variables and lowercase for values. 
For $m\in \N$, let $[m] = \set{1, \dots, m}$. For a random variable $\rv{X}$, let $x\getsr \rv{X}$ to denote that $x$ is chosen according to
$\rv{X}$. Similarly, for a set $\Sa$ let $s\getsr \Sa$ to denote that $s$ is chosen according to the uniform distribution over
$\Sa$. The support of the distribution $D$, denoted $\Supp(D)$, is defined as $\set{u \in \cu : \ppr{D}{u} > 0}$. The
statistical distance between two distributions $P$ and $Q$ over a finite set $\cu$, denoted $\SD(P, Q)$, is defined as $\frac{1}{2}\sum_{u\in \cu}|\ppr{P}{u}-\ppr{Q}{u}|$, which is equal to $\max_{\cs\subset \cu} (\ppr {P}\cs - \ppr{Q}\cs)$.

For a vector $\vec{X}=X_1,...,X_n$ and an index $i \in [n]$, let $X_{<i}$ denote the vector $X_1,...,X_{i-1}$ and $X_{\le i}$ denote the vector $X_1,...,X_{i}$. For a set of indexes $T=\set{i_1,\dots,i_k} \subseteq [n]$ such that $i_1 < i_2<\dots<i_k$, let $X_T$ denote the vector $X_{i_1},\dots,X_{i_k}$. Similarly,  $X_{T,<i}$ denotes the vector  $X_{T \cap \set{1,\dots,i-1}}$. For a function $f$, let $f(\vec{X})=(f(X_1),...,f(X_n))$.

For random variables $\rv{A}$ and $\rv{B}$ we use $\rv{A}|_{\rv{B}=b}$ to denote the distribution of $\rv{A}$ condition on the event $\rv{B}=b$, and $\rv{A}\times\rv{B}$ to denote the product between the marginal distributions of $\rv{A}$ and $\rv{B}$. When $\rv{A}$ is independent from $\rv{B}$ we write $\rv{A}\bot\rv{B}$ to emphasize that this is the case.

\subsection{Interactive Protocols}
A two-party protocol $\Pi = (\Ac, \Bc)$ is a pair of probabilistic interactive Turing
machines. The communication between the Turing machines $\Ac$ and $\Bc$ is carried out in rounds,
where in each round one of the parties is active and the other party is idle. In the $j$-th round of
the protocol, the currently active party $\Pc$ acts according to its partial view, writing some value on
its output tape, and then sending a message to the other party (i.e., writing the message on the
common tape).
The communication transcript (henceforth, the transcript) of a given execution of the protocol
$\Pi = (\Ac, \Bc)$, is the list of messages $m$ exchanged between the parties in an execution of the protocol,
where $m_{1,...,j}$ denotes the first $j$ messages in $m$. A view of a party contains its input, its random
tape and the messages exchanged by the parties during the execution. Specifically, $\Ac$'s view is a
tuple $v_\Ac = (i_\Ac, r_\Ac,m)$, where $i_\Ac$ is $\Ac$'s input, $r_\Ac$ are $\Ac$'s random coins, and $m$ is the transcript of the
execution. Let $\out^\Ac$ denote the output of $\Ac$ in the end of the protocol, and $\out^\Bc$ $\Bc$'s output. Notice that given a protocol, the transcript and the outputs are deterministic function of the joint view $(i_\Ac,r_\Ac,i_\Bc,r_\Bc)$. For a joint view $v$, let $\trans(v)$, $\out^\Ac(v)$ and $\out^\Bc(v)$ be the transcript of the protocol and the parties' outputs determined by $v$.
For a distribution $D$ we denote the distribution over the parties' joint view in a random execution of $\Pi$, with inputs drawn from $D$ by $\Pi(D)$.

A protocol $\Pi$ has $r$ rounds, if for every possible random tapes for the parties, the number of
rounds is exactly $r$. The Communication Complexity of a protocol $\Pi$, denoted as $\CC(\Pi)$ is the length of the transcript of the protocol in  the worst case.

\subsection{Oracle-Aided Protocols}
An oracle-aided two-party protocol $\Pi = (\Ac, \Bc)$ is a pair of interactive Turing machines, where each
party has an additional tape called the oracle tape; the Turing machine can make a query to the oracle by writing a string $q$ on its tape. It then receives a string $ans$ (denoting the answer for this
query) on the oracle tape. An oracle-aided protocol is $\ell$-queries protocol if each party makes at most $\ell$ queries during each run of the protocol. In a  \emph {non-adaptive} oracle-aided protocol, the parties choose their  queries before the protocol starts and before querying the oracle. A  \emph {uniform query} oracle-aided protocol, is a non-adaptive protocol in which the parties queries are chosen uniformly form a predetermined set.

\subsection{Key-Agreement Protocols}
Since we are giving lower bounds, we focus on single bit protocols.
\begin{definition} [key-agreement protocol]\label{def:KAProtocol} Let $0 \leq \gamma$, $\alpha \leq 1$ and $q \in N$. A two-party  boolean output protocol $\Pi = (\Ac, \Bc)$ is a $(q, \alpha, \gamma)$-key-agreement relative to a function family $\FFam$,  if the
	following hold:
	\begin{description}
		\item[Accuracy:] {\sf $\Pi$ has $(1 - \alpha)$-accuracy.}  For every $f \in \FFam$:
		$$\ppr{v\getsr \Pi^f}{\out^\Ac(v)=\out^\Bc(v)} \geq 1-\alpha.$$
		
		\item[ Secrecy:] {\sf $\Pi$ has $(q,\gamma)$-secrecy.}     For every  $q$-query oracle-aided algorithm $\Ec$:
		$$\ppr{f \getsr \FFam, v\getsr \Pi^f}{\Ec^f(\trans(v))=\out^{\Ac}(v)}\leq \gamma.$$
			
	\end{description}

\end{definition}

If $\FFam$ is a  trivial function family (\eg $\FFam$ contains only the identity function), then all correlation between the parties' view is implied by the transcript. Hence, an adversary that on a given transcript $\tau$ samples a random view for $\Ac$ that is
consistent with $\tau$, and outputs whatever $\Ac$ would upon this view, agrees with $\Bc$ with the same
probability as does $\Ac$. This simple argument yields the following fact.

\begin{fact} \label{fact:noITKA}
	For every $0 \leq \alpha \leq 1$ and $0 \leq\gamma < 1-\alpha$, there exists no $(q,\alpha,\gamma)$-key-agreement protocol relative to the trivial family.
\end{fact}

\subsection{Entropy and Information}
The Shannon Entropy of a random variable $\rv{A}$ is defined as $\HH(\rv{A})=\sum_{a\in \Supp(\rv{A})}\ppr{\rv{A}}{a}\log\frac{1}{\ppr{\rv{A}}{a}}$. The conditional entropy of a random variable $\rv{A}$ given $\rv{B}$ is defined as $\HH(\rv{A}|\rv{B})= \mathop{E}_{b \getsr \rv{B}}[\HH(\rv{A}|_{\rv{B}=b)}]$. The following fact is called the chain rule of Shannon Entropy:

\begin{fact}[Chain rule for entropy]\label{fact:HChainRule}
	For a random variable $\vec{A}=\rv{A}_1,...,\rv{A}_n$ the following holds:
	$$\HH(\rv{A}_1,...,\rv{A}_n)=\sum_{i=1}^n\HH(\rv{A}_i|\rv{A}_1,...\rv{A}_{i-1}).$$
\end{fact}

For a function $f$, the $f$-divergence between random variables $\rv{A}$ and $\rv{B}$, denoted as $D_f(\rv{A}, \rv{B})$, is defined as $D_f(\rv{A}, \rv{B})=\sum_{b\in \rv{B}}\pr{\rv{B}=b}f(\frac{\pr{\rv{A}=b}}{\pr{\rv{B}=b}})$.
We use $\MI_f$ as ``mutual information with respect to the $f$-divergence'', $\MI_f(\rv{A};\rv{B}) = D_f( (A, B), (A \times B))$.
The conditional mutual information, $\MI_f(\rv{A};\rv{B}|\rv{C})$ is defined as $\ex{c\getsr  \rv{C}}{\MI_f(\rv{A}|_{\rv{C}=c};\rv{B}|_{\rv{C}=c})}$

For $f(t)=1/2|t-1|$, $\MI_f=\MI_{SD}$ is the statistical distance between the joint distribution to the product, that is, $\MI_f(\rv{A};\rv{B}) = \SD((\rv{A}\rv{B}),(\rv{A}\times\rv{B}))$.

For $f(t)=t\log t$, the f-divergence is called the KL-divergence, and $\MI_f$ (from here denoted as $\MI_{KL}$ or simply $\MI$) is the mutual information
 $\MI(\rv{A};\rv{B})=\HH(\rv{A})-\HH(\rv{A}|\rv{B})$. The mutual information is known to be symmetric, and the following facts are known: 
\begin{fact} [Chain rule for information] \label{fact:IChainRule}
	For random variables $\vec{A}=\rv{A}_1,...,\rv{A}_n$ and $\rv{B}$, 
	$$\MI(\rv{A};\rv{B})=\sum_{i=1}^n\MI(\rv{A}_i;\rv{B}|\rv{A}_1,...,\rv{A}_{i-1}).$$
\end{fact}
\begin{fact}\label{fact:IBounds}
	For every random variables $\rv{A}$ and $\rv{B}$, $0 \leq \MI(\rv{A};\rv{B})\leq \HH(\rv{A}) \leq |\rv{A}|.$
\end{fact}

\begin{fact} [Data processing inequality]\label{fact:DataProcessing}
	Let $\rv{A},\rv{B}$ be random variables, and $f$ a function. Then:
	$\MI(f(\rv{A});\rv{B})\leq \MI(\rv{A};\rv{B})$ and $\HH(f(\rv{A}))\leq \HH(\rv{A})$.
\end{fact}

Lastly, a connection between mutual information and statistical distance is known:
\begin{fact}[Pinsker's inequality]\label{fact:Pinsker}
	$$\ISDC{\rv{A}}{\rv{B}} \leq 2\sqrt{\MI(\rv{A};\rv{B})}.$$
\end{fact}

We will also use the next general lemmas in our proof. The proofs are in \cref{sec:Appendix}. 
\begin{lemma}\label{lemma:ICondAdd}  For every random variables $\rv{A},\rv{B},\rv{C}$ and $\rv{D}$ it holds that 
	$$-\MI(\rv{A};\rv{D}|\rv{C}) \leq \MI(\rv{A};\rv{B}|\rv{C},\rv{D
	})-\MI(\rv{A};\rv{B}|\rv{C})\leq \MI(\rv{A};\rv{D}|\rv{C},\rv{B})$$.
\end{lemma}

The next two lemmas are useful in bounding information by using Bernoulli random variables:
\begin{lemma}\label{lemma:RVToIndicator}
	Let $\rv{J}$ be a Bernoulli random variable, s.t. $\pr{\rv{J}=1}\leq 1/2$. Then $$\HH(\rv{J})\leq \pr{\rv{J}=1}(\log{\frac{1}{\pr{\rv{J}=1}}}+4).$$ 
\end{lemma}

\begin{lemma}\label{lemma:IChainRuleIndi}
	Let $\rv{A},\rv{B},\rv{M}$ and for each $m\in M$ $\rv{E}_m$ be random variables. Let $\rv{J}_m$ be the indicator for the event $\rv{M}=m$, then
	$$\MI(\rv{A};\rv{B}|\rv{M},\rv{E}_M) \leq \sum_{m \in M}\big[I(\rv{A};\rv{B}|\rv{E}_m)+I(\rv{J}_m;\rv{B}|\rv{E}_m,\rv{A})\big].$$
	
\end{lemma}

\subsubsection{Some Useful Facts }

\begin{fact} [Data processing inequality for statistical distance]\label{fact:DataProcessingSD}
	Let $\rv{A},\rv{B}$ be random variables, and $f$ a function. Then:
	$\SD(f(\rv{A}),f(\rv{B}))\leq \SD(\rv{A},\rv{B})$.
\end{fact}
\begin{fact} \label{fact:SDEx}
	Let $\rv{A},\rv{B},\rv{C}$ be random variables. Then:
	\begin{equation*}\SD((\rv{A},\rv{B}),(\rv{A},\rv{C}))=
	 \ex{a\getsr\rv{A}}{\SD(\rv{B}|_{\rv{A}=a},\rv{C}|_{\rv{A}=a})}.
	 \end{equation*}
\end{fact}
\begin{fact} \label{fact:SDRemoveProduct}
	Let $\rv{A},\rv{B},\rv{C}$ be random variables. Then:
	$\SD((\rv{A}\times\rv{B}),(\rv{A}\times\rv{C}))=
	\SD(\rv{B},\rv{C})$.
\end{fact}

\begin{fact}[Hoeffding's inequality\cite{Hoeffding63}]\label{fact:Hoeffding}
	Let $\rv{A}_1, ..., \rv{A}_n$ be independent random variables s.t. $\rv{A}_i \in [0, 1]$ and let $\widehat{\rv{A}} = \frac{1}{n}\Sigma_{i=1}^n \rv{A}_i$.
	It holds that:
	\begin{align*}
	\pr{\widehat{\rv{A}}-\ex{}{\widehat{\rv{A}}}\geq t}\leq e^{-2nt^2}.
	\end{align*}
\end{fact}

\begin{fact}[Jensen's inequality]\label{fact:Jensen}
	Let $f$ be some convex function, and $x_1,...,x_n$ some numbers in $f$'s domain. And let $w_1,...,w_n$ be positive weights such that $\Sigma w_i =1$. Then:
	\begin{align*}
	f(\Sigma w_ix_i) \geq \Sigma w_i f(x_i).
	\end{align*}
	
\end{fact}

The proofs for the next three lemmas are appear in \cref{sec:Appendix}:
\begin{lemma}\label{lemma:SDChainRule}
	Let $\rv{A},\rv{B}$ and $\rv{C}$ be random variables. Then
	\begin{equation*}
	\ex{c\getsr\rv{C}}{\ISD{\rv{A}}{\rv{B}}{\rv{C}=c}} \leq 2\ISDC{\rv{A},\rv{C}}{\rv{B}}.
	\end{equation*}	
\end{lemma}

\begin{lemma}\label{lemma:SDCondAdd}
	Let $\rv{A},\rv{B}$ and $\rv{M}$ be random variables. Then
	\begin{align*}
&\ISDC{\rv{M}}{ \rv{A}}
\leq
	\ex{b\getsr \rv{B}}{\ISD{\rv{M}}{\rv{A}}{\rv{B}=b}}+
	\ISDC{\rv{A}}{\rv{B}}.
	\end{align*}
\end{lemma}

\begin{lemma}\label{lemma:SDVarRep}
	Let $\rv{A},\rv{B}$ and $\rv{M}$ be random variables. Then \begin{align*}
	&\ex{m\getsr \rv{M}}{\ISD{\rv{A}}{\rv{B}}{\rv{M}=m}}
	\leq 2\ex{b\getsr \rv{B}}{\ISD{\rv{A}}{\rv{M}}{\rv{B}=b}}+2\ISDC{\rv{A}}{ \rv{B}}.
	\end{align*}
\end{lemma}
For our proof we need only the following specific case of \cref{lemma:SDVarRep}:
\begin{corollary}\label{cor:SDVarRep}
	Let $\rv{A},\rv{B}$ and $\rv{M}$ be random variables, such that $\rv{A}\bot\rv{B}$. Then \begin{align*}
	&\ex{m\getsr \rv{M}}{\ISD{\rv{A}}{\rv{B}}{\rv{M}=m}}
	\leq 2\ex{b\getsr \rv{B}}{\ISD{\rv{A}}{\rv{M}}{\rv{B}=b}}.
	\end{align*}
\end{corollary}

\section{Uniform-Query Protocols}\label{sec:UnifromProtocols}
In this section, we prove a lower bound on the communication complexity of \emph{uniform-query} key-agreement  protocols.  Recall that an oracle-aided  protocol has \emph{ uniform-queries}, if the  queries made by the parties are uniformly chosen independently from an  (a-priori fixed) domain. 
Our bound is that  an $\ell$-uniform-query protocol secure against  $\ell^2$-query eavesdropper, must have communication complexity  $\Omega(\ell)$. It follows that the uniform-query  protocol of \citet{Merkle82a} (\ie Merkle  puzzle) has optimal   communication complexity (up to a log factor) for such protocols.  
 We prove the bound by exhibiting a reduction from uniform-query key-agreement protocol to (no oracle) protocol for solving  the \textit{set-disjointness problem}.

\begin{definition}[Set-disjointness]
	 Protocol $\Pi=(\Ac,\Bc)$ {\sf solves   set-disjointness    with error $\epsilon$  over distribution $D$ (with support $(\zs)^\ast \times (\zs)^\ast$)},  if 
	
	$$\ppr{\substack{(\setx,\sety) \getsr D \\r_\Ac \getsr \zs, r_\Bc \getsr \zs\\ r_p \getsr \zs}}{(\Ac(\setx;r_\Ac),\Bc(\sety;r_\Bc))(r_\p) = (\setx\cap\sety=\emptyset \land  \setx\cap\sety=\emptyset)} \ge 1- \eps.$$
\end{definition}
Namely, with save but probability $\eps$ over the instance in hand and their private and public randomness, the parties find outs whether their two input sets  intersect.   Our reduction is to solving  set-disjointness  over the distribution below, known to be hard for low complexity protocols.

\begin{definition}[hard distribution for set-disjointness]
For $\ell \in \N$, let \\	 $\q^0_\ell = \{\setx,\sety\subset [\ell] \colon\size{\setx}=\size{\sety}= \floor{\ell/4} \ , \ \setx\cap\sety=\emptyset\}$ and  let\\ $\q^1_\ell = \set{\setx,\sety\subset [\ell]\colon \size{\setx}=\size{\sety}=\floor{\ell/4} , \size{\setx\cap\sety}=1}$. Let $D_\ell^0$ and $ D_\ell^1$ be the uniform distribution over $\q^0_\ell $ and $\q^1_\ell$ respectively, and let $D_\ell=\frac34 \cdot D^0_\ell+\frac14 \cdot D^1_\ell$.
\end{definition}
\citet{razborov1992distributional} has shown that  solving set-disjointness $D_\ell$ with small error require high communication complexity.
 \begin{theorem}[hardness of  $D_\ell$, \cite{razborov1992distributional}]\label{thm:hardnessOfDn}
 Exists $ \epsilon>0$  such that  for  every $\ell \in \N$ and a protocol $\Pi$ that solves set-disjointness over $D_\ell$ with error $\epsilon$,  it holds that $\CC(\Pi) \geq \Omega(\ell)$.
 \end{theorem}
 
 For a finite set $\cs$, let $\FFam_\cs=\set{f:\cs\mapsto\set{0,1}^*}$ be the family of all functions from $\cs$ to binary strings. Our reduction is stated in the following theorem. 
\begin{theorem}[from uniform-query key-agreement protocols to \\set-disjointness]\label{thm:mainUniformRed}
	Assume  exists an $\ell$-uniform-query  $(0, \alpha,\gamma)$-key agreement protocol relative to  $\FFam_\cs$, for some set $\cs$, of  communication complexity $c$. Then  there exists a protocol for solving set-disjointness over $D_\ell$ with $\epsilon$ error and  communication complexity $\frac{2^{15}\cdot \ell^4 \cdot \log{1/\epsilon}}{\size{\cs}^2(1-\alpha-\gamma)^4} \cdot c$.
\end{theorem}

Note that the above theorem holds also for protocols that are only secure against eavesdropper without access to the oracle. 
Combining \cref{thm:hardnessOfDn,thm:mainUniformRed} yields the following bound on the communication complexity of uniform-query  key-agreement  protocols.
\begin{theorem}[Main result for uniform-inputs protocols]\label{thm:mainUniform}
	For any  $\ell$-uniform-query $(q, \alpha,\gamma)$-key agreement protocol  $\Pi$  relative to $\FFam_\cs$, it holds that   $\CC(\Pi) \in \Omega((1-\alpha-\gamma)^4q^2/\ell^3)$.
	\end{theorem}

\begin{proof}
By \cref{thm:hardnessOfDn,thm:mainUniformRed}, protocol  $\Pi$ has communication complexity $\Omega((1-\alpha-\gamma)^4\size{\cs}^2/\ell^3)$. 
By \cref{fact:noITKA}, an eavesdropper that queries all the elements in $\cs$  can guess the key with probability $1-\alpha$. Since \wlg $1-\alpha > \gamma$,  it must hold that $q < \size{\cs}$.  Hence,  $\CC(\Pi) \in \Omega((1-\alpha-\gamma)^4q^2/\ell^3)$.
\end{proof}

The rest of this section is devoted for proving \cref{thm:mainUniformRed}. Assume there exists an $\ell$-uniform-query $(0, \alpha,\gamma)$-key-agreement protocol $\Pi = (\Ac,\Bc)$ relative to the function family $\FFam_\cs$. We  use   $\Pi$ to create a (no-oracle) protocol of about the same communication complexity that finds out the intersection size of parties inputs. We    complete the proof showing that  the latter protocol can be used to solve set-disjointness over the hard distribution $D_\ell$.


Protocol $\Pcom$ below emulates protocol $\Pi$ relative to the family $\FFam_\cs$, in the communication complexity model (where no oracle is given).  The  parties of $\Pcom$  emulate of the random oracle using their shared public randomness interpreted  as (description of  a) function from the function family.

\begin{protocol}[$\Pcom = (\Acom,\Bcom)$]\label{proto:com}~
	\item[$\Acom$'s input:]  an $\ell$-element  set $\setx \subseteq \cs$.  
	\item[$\Bcom$'s input:] an $\ell$-element  set $\sety \subseteq \cs$.  
	\item [Public randomness:] (description of a)  function $f \in \FFam_\cs$.
	\item[Operation:] ~

	$\Acom$ and $\Bcom$ interact in an execution  $(\Ac(\setx,f(\setx)),\Bc(\sety,f(\sety)))$  of $\Pi$, taking the roles of $\Ac$ and $\Bc$ respectively: $\Acom$ acts as $\Ac$ with queries $\setx$ and answers $f(\setx)$, and $\Bcom$ as $\Bc$ with queries $\sety$ and answers $f(\sety)$. At the end of the interaction, $\Acom$ and $\Bcom$ output the outputs of $\Ac$ and $\Bc$ respectively.

\end{protocol} 

We  compare the above protocol to a protocol  that emulates a  run of $\Pi$ \emph{without} using the shared oracle; each party sets  the answers of the oracle using its \emph{private} randomness, and acts accordingly. 

\paragraph{The private-oracle emulation.} In this protocol, each party sample a random function using private randomness. The parties then interact according to $\Pcom$, while treating the private function as the shared oracle.

\begin{protocol}[$\Pdist = (\Adist,\Bdist)$]\label{proto:piI}~
		
		\item[$\Acom$'s input:]  an $\ell$-element  set $\setx \subseteq \cs$.  
		\item[$\Bcom$'s input:] an $\ell$-element  set $\sety \subseteq \cs$.  
			\item [Public randomness:] none.
			 
		\item[Operation:] ~
		\begin{enumerate}
			\item $\Adist$ samples   $g\getsr \FFam_\cs$.	 
			\item $\Bdist$ samples $f\getsr \FFam_\cs$.		
			\item $\Adist$ and $\Bdist$ interact in protocol $(\Ac(\setx,g(\setx)),\Bc(\sety,f(\sety)))$ taking the roles of $\Ac$ and $\Bc$ respectively: $\Adist$ acts as $\Ac$ with queries $\setx$ and answers $g(\setx)$, and $\Bdist$ as $\Bc$ with queries $\sety$ and answers $f(\sety)$. At the end of the interaction, $\Adist$ and $\Bdist$ output the outputs of $\Ac$ and $\Bc$ respectively.
		\end{enumerate}
\end{protocol}

 Let $(\setX,\setY)$ be distributed as the queries of  parties $\Ac$ and $\Bc$ respectively in $\Pi$ (that is, uniform sets in $\cs$ of size $\ell$), and recall that $\Pdist(\setX,\setY)$ and $\Pcom(\setX,\setY)$ denote the parties' joint view  in a random execution of  $\Pdist$ and $\Pcom$ respectively, with inputs drawn from $(\setX,\setY)$.  We first  show that $\Pdist(\setX,\setY)$ is far  from $\Pcom(\setX,\setY)$. Indeed, since $\Pdist$  is a no-oracle protocol (and has no common randomness),  \cref{fact:noITKA} yields that there is an algorithm $\Ec$ such that
\begin{align}
\ppr{v\getsr \Pdist(\setX,\setY)}{\Ec(\trans(v))=\out^{\Adist}(v)}=\ppr{v\getsr \Pdist(\setX,\setY)}{\out^{\Bdist}(v)=\out^{\Adist}(v)}
\end{align}

In contrast, since $\Pcom$ is an emulation of the protocol $\Pi$ with a random oracle, the secrecy of $\Pi$ and the fact that $\Ec$ sees not the common randomness, yields that

\begin{align}
\ppr{v\getsr \Pcom(\setX,\setY)}{\Ec(\trans(v))=\out^{\Acom}(v)}=\ppr{f \getsr F, v\getsr \Pi^f}{\Ec^f(\trans(v))=\out^{\Ac}(v)}\leq \gamma
\end{align}

Finally, since the joint distribution of the outputs of the parties in $\Pcom$ is exactly as in $\Pi$, it holds that
 
 \begin{align}
 \ppr{v\getsr \Pcom(\setX,\setY)}{\out^{\Bcom}(v)=\out^{\Acom}(v)}= \ppr{f \getsr F, v\getsr \Pi^f}{\out^\Ac(v)=\out^\Bc(v)}\geq 1-\alpha
 \end{align}

  It follows that at least one of the two equations  below holds:

\begin{align}\label{eq:accGap}
&\mbox{Agreement gap: }\\\nonumber
&\ppr{v\getsr \Pcom(\setX,\setY)}{\out^{\Bcom}(v)=\out^{\Acom}(v)}-\ppr{v\getsr \Pdist(\setX,\setY)}{\out^{\Bdist}(v)=\out^{\Adist}(v)}\\\nonumber
&\geq (1-\alpha-\gamma)/2
\end{align}

\begin{align}\label{eq:SecGap}
&\mbox{Secrecy gap: }\\\nonumber
& \ppr{v\getsr \Pdist(\setX,\setY)}{\Ec(\trans(v))=\out^{\Adist}(v)}-\ppr{v\getsr \Pcom(\setX,\setY)}{\Ec(\trans(v))=\out^{\Ac}(v)}\\\nonumber
&\geq (1-\alpha-\gamma)/2
\end{align}

Namely,  wither   $\Pcom$ is  significantly  more accurate than  $\Pdist$,  or $\Pcom$ is significantly more secure than protocol $\Pdist$ (or both). We claim that \wlg one can assume that   \cref{eq:accGap} holds (\ie there is agreement gap). Assuming otherwise (\ie \cref{eq:SecGap} holds),  we build a new protocol with inaccurate no-oracle emulation, and then continue the proof assuming \cref{eq:accGap} holds. 

 Consider  protocols $\Pitag=(\Atag,\Btag)$ and $\Pdisttag = (\Adisttag,\Bdisttag)$, in which the parties interact according to $\Pcom$ and $\Pdist$ respectively, but parties $\Btag$ and $\Bdisttag$ output $\neg\Ec(\trans)$.  By the secrecy gap assumption,

\begin{align}
&\ppr{v\getsr  \Pdist(\setX,\setY)}{\Ec(\trans(v))=\out^{\Adist}(v)}-\ppr{v\getsr \Pcom(\setX,\setY)}{\Ec(\trans(v))=\out^{\Acom}(v)}\\\nonumber
&\geq (1-\alpha-\gamma)/2
\end{align}
Hence,
\begin{align*}
\lefteqn{\ppr{v\getsr \Pitag(\setX,\setY)}{\out^{\Btag}(v)=\out^{\Atag}(v)} -\ppr{v\getsr \Pdisttag(\setX,\setY)}{\out^{\Bdisttag}(v)=\out^{\Adisttag}(v)}}\\
&\begin{aligned}
=&(1-\ppr{v\getsr \Pcom(\setX,\setY)}{\Ec(\trans(v))=\out^{\Bcom}(v)})\\
&-(1-\ppr{v\getsr  \Pdist(\setX,\setY)}{\Ec(\trans(v))=\out^{\Adist}(v)})
\end{aligned}\\
&=\ppr{v\getsr \Pdist(\setX,\setY)}{\Ec(\trans(v))=\out^{\Adist}(v)}-\ppr{v\getsr \Pcom(\setX,\setY)}{\Ec(\trans(v))=\out^{\Acom}(v)}\\&\geq (1-\alpha-\gamma)/2.
\end{align*}
That is,  protocol $\Pdisttag$ is less accurate than $\Pitag$ by $(1-\alpha-\gamma)/2$. Namely,  we are exactly in the same situation as if \cref{eq:accGap} holds, but \wrt protocols  $\Pdisttag$ and  $\Pitag$.  From hereafter,  we assume for concreteness that  \cref{eq:accGap} holds \wrt the original protocols $\Pcom$ and $\Pdist$.

\subsection{From Agreement Gap to Set Disjointness}\label{subsec:InAccurate}
Since, by assumption, $\Pdist$ is less accurate than $\Pcom$ in  (\ie \cref{eq:accGap} holds), it is less accurate  for some specific intersection size; when  the parties have  \emph{no} common query, $\Pdist$ behaves just like $\Pcom$, and thus  $\Pdist$ is  (perfectly)  accurate in this case. We exploit this observation to show that the accuracy difference between  the protocols    enables us to  distinguish between disjoint inputs and intersecting inputs, yielding a protocol that solves  set intersection over certain distributions.



For $z\in\set{\Com,\Dist}$  and a joint view $v = (\setx,r_\Ac,\sety,r_\Bc,r_\p)\in\Supp(\Pz)$, let $\x(v)=\setx$ and $\y(v)=\sety$.  For $i \in [\ell]$, let $\Acc{z}{i}$ be the accuracy of $\Lambda_z$ on inputs with intersection size $i$. Namely,
\begin{align*}
\Acc{z}{i} \eqdef \ppr{v\getsr \Pz(\setX,\setY)}{\out^{\Bz}(v)=\out^{\Az}(v)\mid \size{\x(v)\cap\y(v)}=i}.
\end{align*}
Let $\Suc{i}$ be  the accuracy advantage of $\Pcom$ over $\Pdist$ on inputs with intersection size $i$. That is, 
$$ \Suc{i} \eqdef \Acc{\Com}{i}-\Acc{\Dist}{i}$$

A key observation is  that for  some  intersection size, protocol  $\Pcom$ is more accurate than $\Pdist$.  

\begin{claim}\label{claim:goodK}
	$\exists d < \frac{4\ell^2}{\size{\cs}(1-\alpha-\gamma)}$ such that $\Suc{d}\geq (1-\alpha-\gamma)/4$.
\end{claim}
The proof for this claim appears in \cref{sec:Appendix}.

In contrast  to the above claim, if the inputs are \emph{disjoint}  then  there is no agreement gap. That is, we have the following fact.
\begin{claim}\label{eq:goodK}
$\Suc{0} =0$.
\end{claim}
  \begin{proof}
  It is clear  that  for $(\rv{F},\rv{G})\getsr\FFam_\cs^2$ and pair of sets  $\setx\subseteq \cs,\sety \subseteq \cs$ with $\setx\cap\sety = \emptyset$, the distributions of  $(\setx,\sety,\rv{F}(\setx),\rv{F}(\sety))$ and of $(\setx,\sety,\rv{F}(\setx),\rv{G}(\sety))$ are the same. It follows that the distribution $\Pdist|_{\x\cap\y = \emptyset}$ is identical to that of  $\Pcom|_{\x\cap\y = \emptyset}$, meaning that the protocols act the same.	
  \end{proof}

Combining  \cref{claim:goodK} and \cref{eq:goodK} yields   there exists some constant $0< c\leq d$ such that 
\begin{align}
\Suc{c}-\Suc{c-1} \geq \Suc{d}/d \geq \frac{(1-\alpha-\gamma)^2 \cdot \size{\cs}}{16\ell^2}
\end{align}

Hence,
\begin{align}
\frac{(1-\alpha-\gamma)^2\size{\cs}}{16\ell^2}&\leq  \Suc{c}-\Suc{c-1}  \\
&\begin{aligned}
=&\left(\Acc{\Com}{c}-\Acc{\Dist}{c}\right)-\left(\Acc{\Com}{c-1}-\Acc{\Dist}{c-1}\right)
\end{aligned}\nonumber\\
&\begin{aligned}
=\Acc{\Com}{c}-\Acc{\Com}{c-1}+\Acc{\Dist}{c-1}-\Acc{\Dist}{c}
\end{aligned}.\nonumber
\end{align}
Therefore, either

	 \begin{align}\label{case:tilde}\\\nonumber
	&\Acc{\Com}{c}-\Acc{\Com}{c-1}\geq \frac{(1-\alpha-\gamma)^2\size{\cs}}{32\ell^2},
	\end{align}
	or
 \begin{align}\label{case:hat}\\\nonumber
	&\Acc{\Dist}{c-1}-\Acc{\Dist}{c}\geq \frac{(1-\alpha-\gamma)^2\size{\cs}}{32\ell^2}.
	\end{align}

Namely,   at least, one of  protocols $\Pcom$ and $\Pdist$ can be used to distinguish between input of intersection of size $c$ and input of $c-1$ with good  probability. We conclude the proof showing how to use this  ability to solve set-disjointness on the hard distribution $D_\ell$.

\paragraph{The set intersection protocol.}

In the following we assume for concreteness that \cref{case:tilde} holds, where the proof assuming \cref{case:hat} holds  follows analogously by replacing $\Pcom$ with $\Pdist$. Consider the  following protocol for solving set intersection (in the standard communication complexity model). For simplicity, we assume that $\ell$ is a multiple of $4$, and that $\cs=\set{1,\dots,\size{\cs}}$.

\begin{protocol}[$\Piline = (\Aline,\Bline)$]\label{proto:succ}~  
		\item[Parameter:] $k \in N$.
		
		\item[$\Aline$'s input:]  an $\ell/4$-element set  $\setx \subseteq [\ell]$.  
		
		\item[$\Bline$'s input:]  an $\ell/4$-element set    $\sety \subseteq [\ell]$.  
		
		\item [Public randomness:] (description of)  $k$ permutations $\sigma_1,...,\sigma_n$ over $\cs$.
		
		\item[Operation:] ~
		\begin{enumerate}
			\item $\Aline$ sets $\setx'= \setx \cup \set{\ell+1,\ell+2,...,\ell+c-1}\cup \set{2\ell,2\ell+1,...,3\ell-\ell/4-c+1}$ and  $\Bline$ sets $\sety'= \sety \cup \set{\ell+1,\ell+2,...,\ell+c-1}\cup \set{3\ell,3\ell+1,...,4\ell-\ell/4-c+1}$. \label{step:succ:1}
			
			\item $\Aline$ sets $\varcounter=0$.  
			
			\item For $\varj=1$ to $k$: 
				\begin{enumerate}
					\item $\Aline$ and $\Bline$ interact in random execution of   $(\Acom(\sigma_\varj(\setx')),\Bcom(\sigma_\varj(\sety')))$, with fresh randomness, taking the roles of $\Acom$ and $\Bcom$ respectively. Let $\out^{\Acom}$ and $\out^{\Bcom}$ be the parties outputs in the execution.
					\item $\Bline$ sends $\out^{\Bcom}$ to $\Aline$. 
					\item If $\out^{\Acom}=\out^{\Bcom}$, $\Aline$ increases $\varcounter$ by one.
			\end{enumerate}
		\item $\Aline$ informs $\Bline$ whether $\varcounter/k> (\Acc{\Com}{c}+\Acc{\Com}{c-1})/2$. If positive, both parties output zero; otherwise, they output one.
	
		\end{enumerate}
		
\end{protocol}

In the following we analyze the success probability and communication complexity of protocol $\Piline$ for   $k = k^\ast \eqdef \frac{2^{13}\ell^4\log{1/\epsilon}}{\size{\cs}^2(1-\alpha-\gamma)^4}$.

\paragraph{Success probability of $\Piline$.}

We  show that for $k = k^\ast$ it holds that
\begin{align}\label{eq:succ} 
\ppr{\substack{(\setx,\sety) \getsr D_\ell \\r_\Ac \getsr \zs\\ r_\Bc \getsr \zs\\ r_p \getsr \zs}}{(\Aline(\setx;r_\Ac),\Bline(\sety;r_\Bc))(r_\p) = (\setx\cap\sety=\emptyset, \setx\cap\sety=\emptyset)} \ge 1- \eps
\end{align}

We prove that \cref{eq:succ} holds for any fixed  $(\setx,\sety)\in \Supp(D_\ell)$. Fix such a pair $(\setx,\sety)$, and assume \wlg that $\size{\cs}>3\ell/(1-\alpha-\gamma)$ (as otherwise the proof of \cref{thm:mainUniformRed} is immediate).
By this assumption,  it holds that  $c \leq d \leq 3/4\ell$. By construction,  the sets $\setx'$ and $\sety'$ set by the parties in \stepref{step:succ:1} of the protocol, are both  of size $\ell$. Since, by definition,  $(\setx,\sety)$ have at most one shared element, it holds that 
\begin{align}
\size{\setx'\cap\sety'} =
\begin{cases}
c & \setx \cap   \sety \neq  \emptyset\\
c-1,  & \text{otherwise.}\\ 
\end{cases}
\end{align}
It follows that if $\size{\setx'\cap\sety'}=c$ and  $\varcounter/k > (\Acc{\Com}{c}+\Acc{\Com}{c-1})/2$,
then  the protocol  outputs the right answer. Similarly,  this is the case if $\size{\setx'\cap\sety'}=c-1$ and  
$\varcounter/k < (\Acc{\Com}{c}+\Acc{\Com}{c-1})/2$.  Given these observations concerning  the protocol correctness, we conclude the proof by  bounding  the probability that $\varcounter/k$ is far from $\Acc{\Com}{\size{\setx\cap\sety}}$.

\begin{claim}\label{claim:succ} 
Let $\rvcounter$ be the value of $\varcounter$  in a random execution of $\Piline$ on inputs $(\setx,\sety)$. Then for every $\epsilon >0$, $\delta >0$ and $k = \ceil{\log(1/\epsilon)/2\delta^2}$, it holds that

 $\pr{\rvcounter/k-\Acc{\Com}{\size{\setx'\cap\sety'}} > \delta} < \epsilon$ and 
 
  $\pr{\Acc{\Com}{\size{\setx'\cap\sety'}-\rvcounter/k} > \delta} < \epsilon$. 
\end{claim}
\begin{proof}
	
	Since the parties randomly permute their inputs,  for every $j\in [k]$ it holds that $\sigma_j(\setx')$ and $\sigma_j(\sety')$ are random sets  drawn (independently of other iteration) from the distribution $(\setX,\setY)|_{\size{\setX\cap\setY}=\size{\setx'\cap\sety'}}$. Therefore, the probability of the parties to have the same output in each run of $\Pi$ is exactly $\Acc{\Com}{\size{\setx\cap\sety}}$. The  stated bound thus follows by by Hoffeding inequality (\cref{fact:Hoeffding}).
\end{proof}
 Let  $\delta = \frac{(1-\alpha-\gamma)^2\size{\cs}}{2^7\ell^2}$. By \cref{case:tilde}, it holds that \\$\delta < (\Acc{\Com}{c}-\Acc{\Com}{c-1})/2$. Hence,  \cref{claim:succ} yields that protocol $\Piline$  error probability on the input pair $(\setx,\sety)$ for  parameter  $k = k^\ast$  is less than $\epsilon$,  and  \cref{eq:succ} follows.

\paragraph{Communication complexity.} In each iteration of protocol    $\Piline$, the parties run protocol $\Pcom$ and send one additional bit. Since $\CC(\Pcom) = \CC(\Pi)$, for $k = k^\ast$ we get that
\begin{align}\label{eq:cc} 
\CC(\Piline) \leq k(\CC{(\Pcom)}+1) + 1 \leq 4k \cdot \CC{(\Pi)} = \frac{2^{15}\ell^4\log{1/\epsilon}}{\size{\cs}^2(1-\alpha-\gamma)^4}\cdot \CC(\Pi)
\end{align}

\paragraph{Proving \cref{thm:mainUniformRed}.}
The proof of \cref{thm:mainUniformRed}  immediately follows that above observations.
\begin{proof}[Proof of \cref{thm:mainUniformRed}.]
	Fix $k = k^\ast =  \frac{2^{13}\ell^4\log{1/\epsilon}}{\size{\cs}^2(1-\alpha-\gamma)^4}$.  \cref{eq:succ} yields that  protocol $\Piline$ solves set-disjointness over $D_\ell$ with error $\epsilon$, and  \cref{eq:cc} yields that  $\CC(\Piline) \le\frac{2^{15}\ell^4\log{1/\epsilon}}{\size{\cs}^2(1-\alpha-\gamma)^4} \cdot\CC(\Pi)$.
\end{proof}

\section{Two-Messages Non-Adaptive Protocols}\label{sec:ArbitraryProtocols}
In this section we prove a lower bound on the communication complexity of any non-adaptive key agreement protocol that uses only two messages.
We consider protocols with respect to
the family $\FFam_n$ of all functions from $\set{0,1}^n$ to $\set{0,1}^n$.

\begin{theorem}[Main theorem for two-message, non-adaptive protocols]\label{thm:MainArbitrary}
	For any $n \in \nat$,
	the communication complexity of a two-message, non-adaptive, $\ell$-query  $(q,\alpha,\gamma)$-key-agreement protocol relative to $\FFam_n$ is
	at least 
	$$\frac{(1-\alpha-\gamma)^2q}{50^2\ell}-6.$$
\end{theorem}

Fix a two-message, non-adaptive, $\ell$-query protocol $\Pi = (\Ac,\Bc)$. Each execution of the protocol specifies the following:
\begin{itemize}
	\item $\rv{X}$ and  $\rv{Y}$, the queries made by $\Ac$ and $\Bc$, respectively;
	\item $\rv{M}_1,\rv{M}_2$, the messages sent in the two rounds;
	\item $\rv{\out}^\Ac$ and  $\rv{\out}^\Bc$, the outputs of the parties.
\end{itemize}
Where $\rv{M_1}$ is a function (not necessarily deterministic) of $\rv{X}$ and $\rv{F}(\rv{X})$, and $\rv{M}_2$ is a function of  $\rv{Y},\rv{F}(\rv{Y})$ and $\rv{M}_1.$
We define an eavesdropper $\Eve = \Eve_\delta$, where $\delta$ is a parameter we will specify later,
and show that $\Eve$ violates the secrecy of $\Pi$
if $\CC(\Pi)$ is too small.
Loosely speaking, the eavesdropper, which is described below,
queries all ``heavy'' queries and outputs
what $\Bc$ would output given these queries.

\begin{algorithm}[The eavesdropper $\Eve$]\label{alg:Eve}~
		\begin{description}
		\item[Oracle:]  $f\in \FFam_n$.
		\item[Parameter:] $\delta>0$.
		\item[Operation:]   Let $\vec{m} = m_1,m_2$ be the messages exchanged in the protocol.
		\begin{enumerate}
			\item 		
			Query  $f$ on all elements in $\cE_0 \cup \cE_1$, defined as
			
			\begin{equation*}
				\cE_{0} =\set{q \in \zn \colon  \pr{q \in X \cup Y  } \geq \delta}.
			\end{equation*}
and
			\begin{equation*}
				\cE_{1} =\set{q \in \zn \colon  \pr{q \in X \cup Y \medspace \middle| \medspace M_{1} = m_{1} , \rv{F}\big|_{\cE_0} = f\big|_{\cE_0} } \geq \delta}.
			\end{equation*}

\item Sample and output 
	\begin{equation*}
		k \getsr \out^\Bc|_{\rv{M}_{\leq 2} = m_{\leq 2} , \rv{F}\big|_{\cE_0 \cup \cE_1}  = f \big|_{\cE_0 \cup \cE_1}} .
	\end{equation*}
	\end{enumerate}
	
	\end{description}
\end{algorithm}

It does not matter if $\Eve$ asks her queries \emph{during} the protocol's run or \emph{afterwards}.
It is convenient to assume that $\Eve$ asks the queries $\cE_{i-1}$ 
after observing $\rv{M}_{\leq i-1}$
and before the next message is sent. 
In particular, $\cE_0$ denotes the queries that are heavy \emph{before} the messages are sent.
These queries are a function of $\Pi$ itself.

\subsection{Simplifying the Structure of the Protocol}
For our lower bound it is convenient to assume that the protocol 
has two structural properties:
\begin{enumerate}[(1)]
	\item There are no queries that are a-priori heavy, that is, $\cE_0 = \emptyset$.
	\item The secret key chosen by the players is the first bit in $\Bc$'s last query; that is, if $\Bc$'s queries are $\rv{Y}_1,\ldots,\rv{Y}_s$, then the secret key is the first bit of $\rv{Y}_s$.
\end{enumerate}
We show that any key agreement protocol can be transformed into one that has these properties, with minor loss in the parameters. The proofs for the next two lemmas appears in \cref{sec:Appendix}.

\paragraph{Eliminating the a priori heavy queries.}
First we show that if $\cE_0 \neq \emptyset$, we can \emph{fix} the answers to $\cE_0$ in advance, eliminating the need for the players and for $\Eve$ to ask these queries.

\begin{lemma}
	Let $\Pi$ be any $\ell$-query $(q,\alpha,\gamma)$-key-agreement protocol.
	Then there is an $\ell$-query $(q-|\cE_0|,\alpha,\gamma)$-protocol $\Theta$ with the same communication complexity as $\Pi$,
	such that $\Theta$ has no queries that are heavy a priori, that is,
	for each $q \in \set{0,1}^n$, and for any oracle $f \in \FFam_n$,
	\begin{equation*}
		\ppr{\theta^f}{ q \in \rv{X} \cup \rv{Y}} \leq \delta.
	\end{equation*}
	\label{lemma:E0}
\end{lemma}


\paragraph{The key can be $\Bc$'s last query.}
Next we show that we can transform any protocol into one where the secret key is the first bit of $\Bc$'s last query.

\begin{lemma}
	Let $\Pi$ be an $\ell$-query $(q,\alpha,\gamma)$-key-agreement protocol with two messages and communication complexity $C$.
	Then there is an $(\ell+1)$-query $(q,\alpha,\gamma)$-protocol $\Theta$ with two messages and communication complexity $C+1$,
	in which the secret key is the first bit of $\rv{Y}_{\ell+1}$.
	\label{lemma:Y1}
\end{lemma}

\subsection{Proof of the Main Theorem}

We are now ready to prove \cref{thm:MainArbitrary}.
Given a $(q, \alpha, \gamma)$-protocol, we showed in the previous section that we can construct a $(q - |\cE_0|, \alpha, \gamma)$-protocol with one extra query and one extra bit of communication,
which has the two properties we need.
Henceforth, we assume that the two structural properties hold.

The heart of the lower bound is the following lemma, which asserts that the eavesdropper $\Eve$ defined above is able to ask enough queries
so that $\Bc$ has very little advantage over $\Eve$
when it comes to outputting a secret key shared with $\Ac$.

Let $\Pi_{\Eve}^F$ denote the distribution of Eve's view under $\Pi^F$.
Namely, it is the
joint distribution of $(\rv{M}_1, \rv{M}_2, \rv{F}(\rv{\cE}_1))$.
We use $v_E$ to denote a view of Eve drawn from this distribution.

\begin{lemma}\label{claim:mainAsSD} 
	\begin{align}
		\ex{ v_E \getsr \Pi_{\Eve}^{\rv{F}}} { \ISD{\rv{X},\rv{F}(\rv{X})}{ \rv{Y}}{v_E} }
				\leq
		25 \sqrt{ \delta (\CC(\Pi)+5) }.
		\label{eq:EveDep}
	\end{align}

\end{lemma}	
simplicity of notation, here and below
we use 
$\rv{X},\rv{F}(\rv{X}),\rv{Y}|_{v_E}$ to denote $\big(\rv{X},\rv{F}(\rv{X}),\rv{Y} \big)|_{v_E}$
(we condition all the three random
variables not just $\rv{Y}$), and similarly in other cases.
We prove 	\cref{claim:mainAsSD} below, but let us first use it to prove \cref{thm:MainArbitrary}.
	
\begin{proof}[Proof of \cref{thm:MainArbitrary}]

		First, let us fix $\delta$ such that $\Eve$ does not ask more than $q$ queries. 
	Let\footnote{In general, for an $r$-message protocol, we would set $\delta = 2r\ell/q$.} $\delta = 4\ell/q$. 
	Since both $\Ac$ and $\Bc$ ask together at most $2\ell$ queries,
	\begin{equation*}
	2\ell 
	\geq
	\E_{\Pi^{\rv{F}}}\left[
	\left|\rv{X}\cup\rv{Y}\right|
	\right]
	=
	\sum_{q \in\set{0,1}^n}\ppr{\Pi^{\rv{F}}}{q \in \rv{X}\cup \rv{Y}} .
	\end{equation*} 
	Since every heavy-query contributes to the sum at least $\delta$, the size\footnote{Recall that we assumed that $\cE_0 = \emptyset$.
			This assumption caused a loss in parameters,
			so here we need to bound the size of $\cE_0$.} of $\cE_0$ is at most $2\ell/\delta=q/2$. 		
	Similarly, for every $m_1$ and 				$f\big|_{\cE_0}$,
	\begin{equation*}
	2\ell 
	\geq
	\sum_{q \in\set{0,1}^n}\ppr{\Pi^{\rv{F}}}{q \in \rv{X}\cup \rv{Y}\mid \rv{M}_1=m_1, \rv{F}\big|_{\cE_0}=f\big|_{\cE_0}} .
	\end{equation*} 
	So, the size of $\cE_1$ is also at most $q/2$. Overall,
	Eve asks no more than $q$ queries.

	Now, recall that $\out^{\Bc}$ is assumed to be the first bit of $\Bc$'s last query.
	In particular, $\out^{\Bc}$ is a deterministic function of $\rv{Y}$.
	From~\cref{eq:EveDep} and \cref{fact:DataProcessingSD}, 
\begin{align*}
	\ex{ v_E \getsr \Pi_{\Eve}^{\rv{F}}} { \ISD{\rv{X},\rv{F}(\rv{X})}{ \rv{\out}^\Bc}{v_E}}
		\\
		\leq
25 \sqrt{ \delta (\CC(\Pi)+5 )}.
\end{align*}
$\Ac$'s output is a function of her view $(\rv{X}, \rv{F}(\rv{X}), \rv{M}_1, \rv{M}_2)$,
so conditioned on $v_E = (\rv{M}_1, \rv{M}_2, \rv{F}(\rv{\cE}_1))$,
it is a function of $(\rv{X}, \rv{F}(\rv{X}))$.
Using the data processing inequality again, we obtain
\begin{align*}
	\ex{ v_E \getsr \Pi_{\Eve}^{\rv{F}} } { \ISD{\rv{\out}^\Ac}{ \rv{\out}^\Bc}{v_E} }
		\leq
	25 \sqrt{ \delta (\CC(\Pi)+5) }.
\end{align*}
Eve samples her output $\out^\Eve$ from $\rv{\out}^{\Bc}|v_E$.
Therefore,
\begin{align}
\notag
&\ppr{\Pi^{\rv{F}}}{\rv{\out}^\Ac=\rv{\out}^\Bc}
-
\ppr{\Pi^{\rv{F}}}{\rv{\out}^\Ac=\rv{\out}^\Eve} \\\nonumber
& = 
	\ex{ v_E \getsr \Pi_{\Eve}^{\rv{F}}} {  \ppr{\Pi^{\rv{F}}|v_E}{\rv{\out}^\Ac=\rv{\out}^\Bc}
	-
	\ppr{\Pi^{\rv{F}}|v_E}{\rv{\out}^\Ac=\rv{\out}^\Eve}}\\
	& \leq 
	25 \sqrt{ \delta (\CC(\Pi)+5) }.
		\label{eq:Eve_success}
\end{align}
In words,
Eve's probability of guessing $\Ac$'s output is close to $\Bc$'s when $\CC(\Pi)$ is small.

On the other hand, we know that $\Pi$ is $\alpha$-consistent and $\gamma$-secure, so Eve \emph{cannot} have a success probability too close to $\Bc$'s:
By the $\alpha$-consistency of $\Pi$, 
we have $\ppr{\Pi^{\rv{F}}}{\rv{\out}^\Ac=\rv{\out}^\Bc}\geq 1-\alpha$.
By the $\gamma$-secrecy, we have\\
$\ppr{\rv{\Pi}^{\rv{F}}}{\rv{\out}^\Ac=\rv{\out}^\Eve}\leq \gamma$.
Together,
\begin{equation}
	 \ppr{\Pi^{\rv{F}}}{\rv{\out}^\Ac=\rv{\out}^\Bc}-\ppr{\rv{\Pi}^{\rv{F}}}{\rv{\out}^\Ac=\rv{\out}^\Eve}
	 \geq
	1-\alpha-\gamma.
	\label{eq:Pi_success}
\end{equation}
Combining~\eqref{eq:Eve_success} and~\eqref{eq:Pi_success} we see that we must have
\begin{equation*}
\CC(\Pi)\geq\frac{(1-\alpha-\gamma)^2}{25^2\delta}-5.
\end{equation*}
\end{proof}

\subsection{Proving \cref{claim:mainAsSD}}
We prove \cref{claim:mainAsSD} by considering each message separately.
We start with an informal exposition of the proof.
The advantage the players obtain over Eve is encapsulated by the difference between
\begin{itemize}
	\item what $\Ac$ and $\Bc$ learn about the intersection $\rv{X} \cap \rv{Y}$ of their query sets given the transcript \emph{and their queries} $\rv{X}$ or $\rv{Y}$; and
	\item what Eve knows about the intersection $\rv{X} \cap \rv{Y}$ given the transcript and \emph{her} queries $\rv{F}(\cE_1)$.
\end{itemize}
To bound this advantage, we argue that
\begin{enumerate}[I.]
	\item After the first message ($\Ac$'s message),
		all the knowledge that $\Bc$ has about $\Ac$'s queries $\rv{X}$ comes from her first message $\rv{M}_1$. Any advantage he has over Eve comes from what he has learned about the intersection $\rv{X} \cap \rv{Y}$ of their query sets.
		Because $\rv{M}_1$ is short,
		$\Bc$ cannot learn too much about this intersection.
From his point of view,
		the posterior distribution of the intersection given $\rv{M}_1$ remains close to the prior (which is known to Eve).

		To establish this part of the argument we use the language of mutual information.

	\item Similarly, after the second message ($\Bc$'s message),
		all the knowledge that $\Ac$ has gained about $\Bc$'s queries $\rv{Y}$ comes from $\rv{M}_2$ \emph{and} what $\Bc$
		already learned about the intersection $\rv{X} \cap \rv{Y}$ from $\rv{M}_1$.
		In particular, there is a small probability that after seeing $\rv{M}_1$, $\Bc$ has learned too much about the intersection,
		and can use this knowledge to communicate with $\Ac$ securely (as Eve does not know the intersection).

		To deal with this low-probability bad event, we need to switch to 
		the language of statistical distance, 
		and use~\cref{lemma:Hybrid} below.
\end{enumerate}
The following technical lemma is useful in the analysis of the second message, as it 
allows to ignore the knowledge $\Bc$ gained about the intersection in the first message. 
This lemma can be useful in other contexts as well.
Its proof appears in \cref{sec:Appendix}.

\begin{lemma}\label{lemma:Hybrid}
	Let $\rv{A} = \rv{A}_1,\ldots,\rv{A}_n$, let $\rv{T} \subseteq [n]$ and let $\rv{B}$ 
	be random variables.
Let $\rv{Z}$ be a random variable taking values in the set $\mathcal{Z}$,
and let $g : \mathcal{Z} \rightarrow \mathcal{P}([n])$ be a function mapping the domain of $\rv{Z}$ to subsets of $[n]$.
Let
	\begin{align*}
		\epsilon = \ex{z \getsr \rv{Z}}{\ex{t \getsr \rv{T}|_z}{\MI(\rv{A}_t;\rv{B}|\rv{A}_{g(z)},z)}}
		\qquad 
		\text{and}\\
		\qquad
		\delta = \ex{z \getsr \rv{Z}}{\ISD{\rv{A},\rv{B}}{\rv{T}}{z}}.
	\end{align*}
	Then 
	\begin{equation*}
		\ex{z,a_{g(z)} \getsr \rv{Z},\rv{A}_{g(z)}}{\SDC{\rv{A}_\rv{T},\rv{T},\rv{B}|_{z,a_{g(z)}}}{(\rv{A}_\rv{T},\rv{T}|_{z,a_{g(z)}})\times \rv{B}|_{z,a_{g(z)}}}}\leq 2\sqrt{\epsilon}+2\delta.
	\end{equation*}
\end{lemma}


\subsection*{Analyzing the first message.}

We start by proving that in expectation, the first message does not create too much dependence between the players' views:
\begin{claim}\label{claim:oneMessage}
	The following statements hold after seeing $\Ac$'s message:
	\begin{enumerate}
		\item $\Ac$'s view remain independent of $\Bc$'s queries: $\MI(\rv{X},\rv{F}(\rv{X});\rv{Y}|\rv{M}_1) = 0$.
		\item The same holds conditioned on Eve's queries: $\MI(\rv{X},\rv{F}(\rv{X});\rv{Y}|\rv{M}_1,\rv{F}(\rv{\cE}_1)) = 0$.
		\item Not much dependence is created between $\Bc$'s view and $\Ac$'s queries:\\ $\MI(\rv{Y},\rv{F}(\rv{Y});\rv{X}|\rv{M}_1) \leq \delta|\rv{M}_1|$.
	\end{enumerate}
\end{claim}

\begin{proof}[Proof for \cref{claim:oneMessage}]
The proof of the first item:
\begin{align*}
	0 \leq \MI(\rv{X},\rv{F}(\rv{X});\rv{Y}|\rv{M}_1)
	&
	\leq \MI(\rv{X},\rv{F}(\rv{X}),\rv{M}_1;\rv{Y})
	\tag{Chain rule}
	\\
	&
	= \MI(\rv{X},\rv{F}(\rv{X});\rv{Y})
	\tag{Since $\rv{M}_1$ is a function of $\rv{X},\rv{F}(\rv{X})$}
	\\
	&=0 .\tag{Because $\rv{Y} \bot (\rv{X}, \rv{F}(\rv{X}))$}
	\end{align*}

The proof of the second item:
			\begin{align*}
	0 \leq \MI(\rv{X},\rv{F}(\rv{X});\rv{Y}|\rv{M}_1, \rv{F}(\cE_1))
	&
	\leq \MI(\rv{X},\rv{F}(\rv{X}),\rv{M}_1, \rv{F}(\cE_1);\rv{Y})
	\tag{Chain rule}
	\\
	&
	= \MI(\rv{X},\rv{F}(\rv{X}), \rv{F}(\cE_1);\rv{Y})
	\tag{Since $\rv{M}_1$ is a function of $\rv{X},\rv{F}(\rv{X})$}
	\\
		&
	\leq \MI(\rv{X},\rv{F};\rv{Y})
	\tag{Data processing}
	\\
	&=0 .\tag{Because $\rv{Y} \bot (\rv{X}, \rv{F})$}
	\end{align*}

	To prove the third item, we first show that all the ``secret information'' $\Bc$ has about $\rv{X}$ after seeing $\rv{M}_1$ --- that is, the dependence between his view and $\rv{X}$ given $\rv{M}_1$ ---  comes from the intersection between $\Ac$ and $\Bc$'s sets. 	

	Let $\rv{T} := \set{i \st \rv{X}_i \in \rv{Y}}$ be the indexes of the intersection queries.
	
	\begin{claim}
		\label{claim:firstMintersection}
		$\MI(\rv{Y},\rv{F}(\rv{Y});\rv{X}|\rv{M}_1)\leq\MI(\rv{M}_1;\rv{F}(\rv{X}_\rv{T})|\rv{T},\rv{X})$.
	\end{claim}
\begin{proof}
	\begin{align*}
	\MI(\rv{Y},\rv{F}(\rv{Y});\rv{X}|\rv{M}_1)
	&
	=\MI(\rv{Y},\rv{F}(\rv{Y});\rv{X}|\rv{M}_1)-\MI(\rv{Y},\rv{F}(\rv{Y});\rv{X})
	\tag{Because $\rv{X} \bot (\rv{Y}, \rv{F}(\rv{Y}))$}
	\\
	&
	\leq\MI(\rv{M}_1;\rv{Y},\rv{F}(\rv{Y})|\rv{X})
	\tag{\cref{lemma:ICondAdd}}
	\\
	&\leq
	\MI(\rv{M}_1;\rv{T},\rv{F}(\rv{X}_\rv{T}),\rv{Y},\rv{F}(\rv{Y})|\rv{X)}
	\\
	&=
	\MI(\rv{M}_1;\rv{T},\rv{F}(\rv{X}_\rv{T})|\rv{X})+\MI(\rv{M}_1;\rv{Y},\rv{F}(\rv{Y})|\rv{X},\rv{T},\rv{F}(\rv{X}_\rv{T}))
	\tag{Chain rule}.
\end{align*}
The second term is 0: because $\rv{M}_1$ is a function of $\rv{X}, \rv{F}(\rv{X})$, we have
\begin{align*}
&
	\MI(\rv{M}_1;\rv{Y},\rv{F}(\rv{Y})|\rv{X},\rv{T},\rv{F}(\rv{X}_\rv{T})) \\
	& \leq \MI(\rv{F}(\rv{X});\rv{Y},\rv{F}(\rv{Y})|\rv{X},\rv{T},\rv{F}(\rv{X}_\rv{T})) \tag{Data processing}
	\\
	&=
	\MI(\rv{F}(\rv{X});\rv{Y}|\rv{X},\rv{T},\rv{F}(\rv{X}_\rv{T}))+
	\MI(\rv{F}(\rv{X});\rv{F}(\rv{Y})|\rv{X},\rv{T},\rv{F}(\rv{X}_\rv{T}), \rv{Y})\tag{Chain rule}
	\\
	&\leq
		\MI(\rv{F}(\rv{X}),\rv{F}(\rv{X}_\rv{T});\rv{Y}|\rv{X},\rv{T})+
	\MI(\rv{F}(\rv{X});\rv{F}(\rv{Y})|\rv{X},\rv{T},\rv{F}(\rv{X}_\rv{T}), \rv{Y})\tag{Chain rule}	
	\\
	&=0+\MI(\rv{F}(\rv{X});\rv{F}(\rv{Y})|\rv{X},\rv{T},\rv{F}(\rv{X}_\rv{T}), \rv{Y})\tag{$(\rv{X},\rv{Y},\rv{T}) \bot \rv{F}$}
	\\
	&=\MI(\rv{F}(\rv{X}\setminus\rv{X}_\rv{T});\rv{F}(\rv{Y}\setminus\rv{X}_\rv{T})|\rv{X},\rv{T},\rv{F}(\rv{X}_\rv{T}), \rv{Y})
	\\
	&=0 .\tag{Since $F$ is a random function and $(\rv{X}\setminus\rv{X}_\rv{T})\cap(\rv{Y}\setminus\rv{X}_\rv{T})=\emptyset$}
\end{align*}
Bound the first term:
\begin{align*}
	&
	\MI(\rv{M}_1;\rv{T},\rv{F}(\rv{X}_\rv{T})|\rv{X})
	\\
	&=
	\MI(\rv{M}_1;\rv{T}|\rv{X})+\MI(\rv{M}_1;\rv{F}(\rv{X}_\rv{T})|\rv{T},\rv{X})
	\tag{Chain rule}
	\\
	&\leq
	\MI(\rv{M}_1;\rv{Y}|\rv{X})+\MI(\rv{M}_1;\rv{F}(\rv{X}_\rv{T})|\rv{T},\rv{X})
	\tag{Data processing: $\rv{T}$ is a function of $\rv{Y}$ given $\rv{X}$}
	\\
	&=\MI(\rv{M}_1;\rv{F}(\rv{X}_\rv{T})|\rv{T},\rv{X})
	\tag{$\rv{M}_1 \bot \rv{Y} | \rv{X}$}.
\end{align*}
\end{proof}

Next, we bound the information $\rv{M}_1$ conveys about $\rv{F}(\rv{X}_{\rv{T}})$,
using the fact that every element in $\rv{X}$ is in the intersection only with small probability (less than $\delta$). The proof of the claim is similar to the proof of 
Shearer's inequality and appears in \cref{sec:Appendix}. 

	\begin{claim}\label{claim:deltaM1}
		$\MI(\rv{M}_1;\rv{F}(\rv{X}_\rv{T})|\rv{T},\rv{X})\leq \delta|\rv{M}_1|$.
	\end{claim}

The proof of the third item is complete.
\end{proof}

\subsection*{Analyzing the second message.} 
We now want to show that the second message also does not create much dependence between $\Ac$ and $\Bc$'s views.
As with \cref{claim:oneMessage} for the first message, we first want to show that all the  dependence between $\Ac$'s view and $\Bc$'s queries
comes from $\Bc$'s message, and that this dependence goes through the intersection between $\Ac$ and $\Bc$'s queries and what the players learn about the intersection from the transcript.  This is done by the next claim.
%
Let 
$$\rv{T}_1:=\set{i \st \rv{Y}_i \in \rv{X} \setminus \rv{\cE}_1}.$$
In words, it is the set of the indices of $\Bc$'s queries in the intersection
that were not queried by Eve.
Recall that $\Pi_{\Eve}^{\rv{F}}$ is the distribution of Eve's view, which includes $\rv{M}_1, \rv{M}_2$ and $\rv{F}(\rv{\cE}_1)$.
Let 
$\rv{B}_E=(\rv{M}_1,\rv{Y},\rv{F}(\rv{Y}\cap \cE_1))$. 

\begin{claim}
	\label{claim:M2ReduceToIntersecton}
	\begin{equation*}
\ex{ v_E \getsr \Pi_{\Eve}^{\rv{F}} }{\ISD{\rv{X},\rv{F}(\rv{X})}{\rv{Y}}{v_E}}
\leq
4\ex{b_E\getsr\rv{B}_E}{\ISD{\rv{T}_1,\rv{F}(\rv{Y}_{\rv{T}_1})}{\rv{M}_2}{b_E}}
.
\end{equation*}
\end{claim}

The proof for the claim appears in \cref{sec:Appendix}.

Now we left to show that on average, $\Bc$'s message cannot convey too much information about the intersection queries and their answers, as we did in \cref{claim:deltaM1} for the first message. Specifically, we want to bound 
\begin{align*}
	\ex{b_E\getsr\rv{B}_E}{\ISD{\rv{T}_1,\rv{F}(\rv{Y}_{\rv{T}_1})}{\rv{M}_2}{b_E}}.	
\end{align*}
It would be easier if $\Bc$ knew \emph{nothing} about the intersection
(i.e.~$\rv{M}_2$ was independent of $\rv{T}_1$ given $\rv{M}_1$). But this is not the case, as $\Bc$ can learn some info from $\Ac$'s message.
However, from \cref{claim:oneMessage}, we know that he does not learn a lot, and his message does not strongly depend on the intersection.
Formally, 
%

\begin{claim}
\label{clm:2ndEnd}
	\begin{align*}
	\ex{b_E\getsr\rv{B}_E}{\ISD{\rv{T}_1,\rv{F}(\rv{Y}_{\rv{T}_1})}{\rv{M}_2}{b_E}}
	 \leq
	6\sqrt{\delta(|\rv{M}_1|+|\rv{M}_2|+5)}.
	\end{align*}
\end{claim} 

The two claims above complete the proof of \cref{claim:mainAsSD}.

\begin{proof}
By definition of $\rv{B}_E$,
		\begin{align*}
			\E_{b_E\getsr\rv{B}_E} & \ISD{\rv{T}_1,\rv{F}(\rv{Y}_{\rv{T}_1})}{\rv{M}_2}{b_E}
			\begin{aligned}
			=\E_{b_E\getsr\rv{B}_E}\ISD{\rv{T}_1,\rv{F}(\rv{Y}_{\rv{T}_1})}{\rv{M}_2}{m_1,y,f(e_1\cap y)} .
			\end{aligned} 
	\end{align*}		
By \cref{lemma:Hybrid},  it is enough to show:
	\begin{enumerate}[(1)]
		\item \begin{align*}
\ex{m_1,y\getsr\rv{M}_1,\rv{Y}}{\ISD{\rv{F}(\rv{Y}),\rv{M}_2}{\rv{T}_1}{m_1,y}}
\leq2\sqrt{\delta |\rv{M}_1|}.
		\end{align*}

		\item $\ex{m_1,y\getsr \rv{M}_1,\rv{Y}}{\ex{t\getsr \rv{T}_1|m_1,y}{\MI(\rv{F}(y_t);\rv{M}_2|m_1,y,\rv{F}(e_1 \cap y))}} \leq \delta(|\rv{M}_1|+|\rv{M}_2|+5)$.
	\end{enumerate}

	The proof of the first item is (which is similar to the analysis of the first message):
	\begin{align*}
	&
	\ex{m_1,y\getsr\rv{M}_1,\rv{Y}}{\ISD{\rv{F}(\rv{Y}),\rv{M}_2}{\rv{T}_1}{m_1,y}}
	\\
	&\leq
	 2\sqrt{\MI(\rv{F}(\rv{Y}),\rv{M}_2;\rv{T}_1|\rv{M}_1,\rv{Y})}\tag{\cref{fact:Pinsker}}
	 \\
	&=
	2\sqrt{\MI(\rv{F}(\rv{Y});\rv{T}_1|\rv{M}_1,\rv{Y})}\tag{$\rv{M}_2$ is a function of $\rv{Y},\rv{F}(\rv{Y}),\rv{M}_1$}
	\\
	&\leq
	2\sqrt{\MI(\rv{F}(\rv{Y});\rv{X}|\rv{M}_1,\rv{Y})}\tag{$\rv{T}_1$ is a function of $\rv{X},\rv{Y}$ and $\rv{M}_1$}
	\\
	&\leq
	 2\sqrt{\MI(\rv{Y},\rv{F}(\rv{Y});\rv{X}|\rv{M}_1)}\tag{Chain rule}
	 \\
	 &\leq
	  2\sqrt{\delta |\rv{M}_1|}\tag{\cref{claim:oneMessage}}.
	\end{align*}
	To bound the second item we use a similar argument to the proof of \cref{claim:deltaM1}.
	The proof is more complicated here, because when we condition on $\rv{M}_1$ and on Eve's queries, the answers of the oracle $\rv{F}$ are no longer independent of each other (e.g., $\Ac$ could send the XOR of the answers to her queries).
	Nevertheless, because not much information was revealed about the oracle's answers, not much dependence is created between them. 
	The proof consists of two steps. First, we show that this term is bounded by $\delta\size{\rv{M}_2}$, plus the dependency between the answers, created by the first message and Eve's queries (\cref{claim:deltaM2}). Next, we bound this dependency (\cref{claim:depEve}).
	\begin{claim}\label{claim:deltaM2}
		\begin{align*}
		&\ex{m_1,y\getsr \rv{M}_1,\rv{Y}}{\ex{t\getsr \rv{T}_1|m_1,y}{\MI(\rv{F}(y_t);\rv{M}_2|m_1,y,\rv{F}(e_1 \cap y))}}\\ &\leq \delta|\rv{M}_2|+\delta\ex{y\getsr \rv{Y}} {\sum_{i}\MI(\rv{F}(y_i);\rv{F}(y_{<i})|\rv{M}_1,y,\rv{F}(\rv{\cE}_1\cap y))}.
		\end{align*}
	\end{claim}
The proof for \cref{claim:deltaM2} is similar to the proof of \cref{claim:deltaM1} and appears in \cref{sec:Appendix}.

	\begin{claim}\label{claim:depEve}
	\begin{align*}
\ex{y\getsr \rv{Y}}{ \sum_{i}\MI(\rv{F}(y_i);\rv{F}(y_{<i})|\rv{M}_1,y,\rv{F}(\rv{\cE}_1\cap y))}\leq \size{\rv{M}_1}+5.
	\end{align*}
\end{claim}
\begin{proof}
	For every $m \in \Supp(\rv{M}_1)$, let $\cE(m)$ be the set of queries Eve asks after seeing the message $m$. 
	By \cref{lemma:IChainRuleIndi}
	(recall that $\rv{J}_m$ is the indicator for the event $\rv{M}=m$),
	\begin{align*}
	&\ex{y\getsr \rv{Y}}{ \sum_{i}\MI(\rv{F}(y_i);\rv{F}(y_{<i})|\rv{M}_1,y,\rv{F}(\rv{\cE}_1\cap y))}\\
	&\begin{aligned}
		\leq\E_{y\getsr \rv{Y}} \sum_{i}\sum_{m \in \rv{M}_1}\big[\MI(\rv{F}(y_i); \rv{F}(y_{<i})|y,\rv{F}(\rv{\cE}(m) \cap y))\\
			+\MI(\rv{F}(y_i);\rv{J}_m|y,\rv{F}(\rv{\cE}(m) \cap y), \rv{F}(y_{<i}))\big]
	\end{aligned}\tag{\cref{lemma:IChainRuleIndi}}
	\end{align*}
For every $m,y,i$, 
by the structure of $\rv{F}$, and since $\rv{F}(\rv{\cE}(m) \cap y)$ is a fixed set,
we have \\$\MI(\rv{F}(y_i); \rv{F}(y_{<i})|y,\rv{F}(\rv{\cE}(m) \cap y))=0$.
Thus, 
	\begin{align*}
	&
\begin{aligned}
\leq\E_{y\getsr \rv{Y}} \sum_{i}\sum_{m \in \rv{M}_1}\big[\MI(\rv{F}(y_i); \rv{F}(y_{<i})|y,\rv{F}(\rv{\cE}(m) \cap y))\\
+\MI(\rv{F}(y_i);\rv{J}_m|y,\rv{F}(\rv{\cE}(m) \cap y), \rv{F}(y_{<i}))\big]
\end{aligned}\\
	&=\ex{y\getsr \rv{Y}}{ \sum_{i}\sum_{m \in \rv{M}_1}\MI(\rv{F}(y_i);\rv{J}_m|y,\rv{F}(\rv{\cE}(m) \cap y), \rv{F}(y_{<i}))}\\
	&=\ex{y\getsr \rv{Y}}{\sum_{m \in \rv{M}_1}\MI(\rv{F}(y);\rv{J}_m|y,\rv{F}(\rv{\cE}(m) \cap y))}\tag{Chain rule}\\
	&\leq\sum_{m \in \rv{M}_1}\HH(\rv{J}_m). \tag{\cref{fact:IBounds}} 
	\end{align*}
	There is at most one $m'$ such that $\pr{\rv{M}_1=m'}\geq 1/2$, hence,
	\begin{align*}
	&\sum_{m \in \rv{M}_1}\HH(\rv{J}_m)\\
	&\leq 1+\sum_{m \in \rv{M}_1}\pr{\rv{M}_1=m}\left(- \log\left(\pr{\rv{M}_1=m}\right)+4\right)
	\tag{\cref{lemma:RVToIndicator}}\\
	&=\HH(\rv{M}_1)+5.
	\end{align*}
\end{proof}

The proof of \cref{clm:2ndEnd} is complete.
\end{proof}

\subsection{Remarks}

\paragraph{Adaptive Protocols.}
While we believe that the eavesdropper  $\Eve$  we defined above should 
allow us to prove lower bounds for every non-adaptive protocol, 
$\Eve$ will not work for adaptive protocol, even if she can choose the sets adaptively as well. \cref{proto:Adaptive} is an example of a one-message protocol with only $O(\log(\ell))$ communication, but without any heavy query (for every $\delta > 1/\ell$). 
Specifically, $\Eve$ will not make any query,
and can not, therefore, break the protocol.
Notice, however, that every one-message protocol can be broken trivially by simulating $\Bc$, so this protocol is not secure.

\begin{protocol}\label{proto:Adaptive}~
	\begin{description}
		\item[Parameters:]$n$, $\ell=2^{n/2}$
		\item[Common functions:]$f,g:\set{0,1}^n\rightarrow\set{0,1}^n$
		\begin{enumerate}
			\item $\Ac$ choses a random string $x \in \set{0,1}^n$ and queries $x, f(x),...,f^{\ell-1}(x)$ and $g(f^{i-1}(x))$ for a random index $i \in [\ell]$.
			\item $\Bc$ choses a random string $y \in \set{0,1}^n$ and queries $y, f(y),...,f^{\ell-1}(y)$ and\\ $g(y),...,g(f^{\ell-1}(y))$.
			\item $\Ac$ sends $M_1 = g(f^{i-1}(x))$ to $\Bc$, and outputs $f^{i-1}(x)$.
			\item If there is $j\in [\ell]$ so that $g(f^{j-1}(y))=M_1$
			then $\Bc$ outputs $f^{j-1}(y)$. Otherwise, $\Bc$ aborts.
		\end{enumerate}
		
	\end{description}.
\end{protocol}

\paragraph{Constant Rounds Protocols}
We failed to continue the proof for 
multi-message protocol.
The main reason is that we were not able to deal with the dependency caused by $\Eve$'s queries. In two-message protocol, Eve's only asks queries after the first message, which depends only on $\Ac$'s view. We show here that conditioning on Eve's view in this case, cannot add too much dependency between $\Ac$ and $\Bc$. However, in protocols with more messages, the queries of Eve depend on the view of both sides, and conditioning on Eve's view can potentially make the dependency more significant.

\ifdefined\IsAnon
\else
\section*{Acknowledgement} 
We thank Yuval Ishai for challenging us with this intriguing question, and  Omer Rotem for very useful discussions.   
\fi

\bibliographystyle{abbrvnat}
\bibliography{crypto} 
\appendix

\section{Merkle's Puzzles}\label{sec:introMP}
For completeness, we briefly  describe here the Merkle Puzzles protocol \cite{Merkle87}. Let $\cs$ be a set of size $\ell^2$, and $\FFam_\cs=\set{f:\cs\mapsto\set{0,1}^{2\log{\size{\cs}}}}$ be the family of all functions from $\cs$ to binary strings of length $2\log{\size{\cs}}$.

\begin{protocol}[Merkle's Puzzles protocol $\Pi = (\Ac,\Bc)$]\label{proto:Merkle}~
	\item [Oracle:]  $f \in \FFam_\cs$.
	\item[Operation:] ~	
	
	\begin{enumerate}
		
		\item $\Ac$ samples uniformly and independently $\ell$ elements $x_1,...,x_\ell\in \cs$, and sets $a_1=f(x_1),...,a_\ell=f(x_\ell)$.

		$\Bc$ samples uniformly and independently $\ell$ elements  $y_1,...,y_\ell\in \cs$, and set  $b_1=f(y_1),...,b_\ell=f(y_\ell)$. 
		
		\item $\Ac$ sends $a_1,...,a_\ell$ to $\Bc$.
		\item $\Bc$ looks for   indices $i,j\in [\ell]$ with  $a_i = b_i$. If no such indices exists,  it  aborts.
		
		\item $\Bc$ sends $i$ to $\Ac$.
		
		\item $\Ac$ outputs $x_i$ and $\Bc$ outputs $y_j$.
	\end{enumerate}
\end{protocol} 

Since each party samples $\ell = \sqrt{\size{\cs}}$ uniform random element from $\cs$, by the birthday paradox they have a common element (\ie collision) with constant probability. By construction, the parties out  the same collision, if such exists. On the other hand, from an attacker point of view the collision is a random element of $\cs$, and therefore she cannot find it with good probability without querying a constant fraction of the element of $\cs$, namely by making $\Theta(\ell^2)$ queries.  

Note that Merkle Puzzles is non-adaptive, uniform-queries, two-message protocol with near linear communication, and therefore shows that our two lower bounds (\cref{thm:mainUniformIntro,thm:MainArbitraryIntro}) are tight.

\section{Missing Proofs}\label{sec:Appendix}

\begin{proof}[Proof of \cref{lemma:ICondAdd}]

	\begin{align*}
	\MI(\rv{A};\rv{B}|\rv{C},\rv{D})-I(\rv{A};\rv{B}|\rv{C})&\\
	&=\HH(\rv{A}|\rv{C},\rv{D})-\HH(\rv{A}|\rv{B},\rv{C},\rv{D})-[\HH(\rv{A}|\rv{C})-\HH(\rv{A}|\rv{B},\rv{C})]\\
	&=\HH(\rv{A}|\rv{C},\rv{D})-\HH(\rv{A}|\rv{C})-[\HH(\rv{A}|\rv{B},\rv{C},\rv{D})-\HH(\rv{A}|\rv{B},\rv{C})]\\
	&=\MI(\rv{A};\rv{D}|\rv{C},\rv{B})-\MI(\rv{A};\rv{D}|\rv{C})\\
	\end{align*}
	The inequalities hold by the fact that mutual information is always positive.
\end{proof}

\begin{proof}[Proof of \cref{lemma:RVToIndicator}]

	\begin{align*}
	\HH(\rv{J})= &\pr{\rv{J}=1}\log{\frac{1}{\pr{\rv{J}=1}}}+\pr{\rv{J}=0}\log{\frac{1}{\pr{\rv{J}=0}}} \\
	&\leq \pr{\rv{J}=1}\log{\frac{1}{\pr{\rv{J}=1}}}+\log{\frac{1}{1-\pr{\rv{J}=1}}}
	\end{align*}
	Let $f(x) = log{\frac{1}{1-x}} -4x$.
	We need to show that $f(x)\leq 0$  for all $0\leq x\leq 1/2 $. $f(0)=0$, therefore it is enough to show that $f'(x)\leq 0$.  
	\begin{align*}
	f'(x)= &\frac{1}{\ln 2}\frac{1}{1-x}-4\\
	&\leq 2\frac{1}{1-x}-4 \leq 4-4 = 0\tag{$0\leq x\leq 1/2$}
	\end{align*}
\end{proof}

\begin{proof}[Proof of \cref{lemma:IChainRuleIndi}]

	\begin{align*}
	\MI(\rv{A};\rv{B}|\rv{M},\rv{E}_M)
	&=\sum_{m\in \rv{M}}\pr{\rv{M}=m}\MI(\rv{A};\rv{B}|\rv{M}=m,\rv{E}_m)
	\\
	&=
	\sum_{m\in \rv{M}}\pr{\rv{J}_m=1}\MI(\rv{A};\rv{B}|\rv{J}_m=1,\rv{E}_m)
	\\
	&\begin{aligned}
	\leq \sum_{m\in \rv{M}}\big[&\pr{\rv{J}_m=1}\MI(\rv{A};\rv{B}|\rv{J}_m=1,\rv{E}_m)\\
	&+\pr{\rv{J}_m=0}\MI(\rv{A};\rv{B}|\rv{J}_m=0,\rv{E}_m)\big]
	\end{aligned}
	\tag{Because $\MI$ is non-negative}
	\\
	&=
	\sum_{m\in \rv{M}}\MI(\rv{A};\rv{B}|\rv{J}_m,\rv{E}_m)
	\\
	&\leq
	\sum_{m\in \rv{M}}\MI(\rv{A},\rv{J}_m;\rv{B}|\rv{E}_m)
	\tag{Chain rule}
	\\
	&=
	\sum_{m \in \rv{M}}\big[\MI(\rv{A};\rv{B}|\rv{E}_m)+\MI(\rv{J}_m;\rv{B}|\rv{E}_m,\rv{A})\big]
	\tag{Chain rule}
	\end{align*}
\end{proof}

\begin{proof}[Proof of \cref{lemma:SDChainRule}]
	\begin{align*}
	&\ex{c \getsr \rv{C}}{\SDC{\rv{A},\rv{B}|_{\rv{C}=c}}{\rv{A}|_{\rv{C}=c}\times \rv{B}|_{\rv{C}=c}} }\\
	&\leq\ex{c \getsr \rv{C}} { \SDC{\rv{A},\rv{B}|_{\rv{C}=c}}{\rv{A}|_{\rv{C}=c}\times \rv{B}} +\SDC{\rv{A}|_{\rv{C}=c}\times \rv{B}}{\rv{A}|_{\rv{C}=c}\times \rv{B}|_{\rv{C}=c}}}\tag{Triangle inequality}\\
	&=\ex{c \getsr \rv{C}} {\SDC{\rv{A},\rv{B}|_{\rv{C}=c}}{\rv{A}|_{\rv{C}=c}\times \rv{B}} +\SDC{ \rv{B}}{ \rv{B}|_{\rv{C}=c}}}\tag{\cref{fact:SDRemoveProduct}}\\
	&\leq\E_{c \getsr \rv{C}} \left[ \SDC{\rv{A},\rv{B}|_{\rv{C}=c}}{\rv{A}|_{\rv{C}=c}\times \rv{B}} +\SDC{ \rv{A}|_{\rv{C}=c}\times\rv{B}}{ \rv{A},\rv{B}|_{\rv{C}=c}}\right]\tag{Data procesing}\\
	&=2\E_{c \getsr \rv{C}} \left[ \SDC{\rv{A},\rv{B}|_{\rv{C}=c}}{\rv{A}|_{\rv{C}=c}\times \rv{B}} \right]\\
	&=2\ISDC{\rv{A},\rv{C}}{\rv{B}}
	\end{align*} 
\end{proof}

\begin{proof}[Proof of \cref{lemma:SDCondAdd}]
	\begin{align*}
	&\SDC{\rv{M}\times \rv{A}}{\rv{M},\rv{A}}\\
	& \leq \ISDCR{\rv{M},\rv{B}}{\rv{A}} \tag{Data processing}\\
	& =  \ex{b \getsr B}{\SDC{\rv{M}|_{\rv{B}=b} \times \rv{A}}{\rv{A},\rv{M}|_{\rv{B}=b}}}\tag{\cref{fact:SDEx}}\\
	&\leq\ex{b\getsr B}{\SDC{\rv{M}|_{\rv{B}=b} \times \rv{A}}{ \rv{M}|_{\rv{B}=b}\times \rv{A}|_{\rv{B}=b} }+\SDC{\rv{M}|_{\rv{B}=b}\times \rv{A}|_{\rv{B}=b} }{\rv{A},\rv{M}|_{\rv{B}=b}}}\tag{Triangle inequality}\\
	&=\ex{b\getsr B}{\SDC{ \rv{A}}{ \rv{A}|_{\rv{B}=b} }+\SDC{\rv{M}|_{\rv{B}=b}\times \rv{A}|_{\rv{B}=b} }{\rv{A},\rv{M}|_{\rv{B}=b}}}\tag{\cref{fact:SDRemoveProduct}}\\
	&=\SDC{\rv{A} \times \rv{B}}{\rv{A},\rv{B}} +\E_{b\getsr \rv{B}}\SDC{\rv{M}|_{\rv{B}=b}\times \rv{A}|_{\rv{B}=b} }{\rv{A},\rv{M}|_{\rv{B}=b}}\tag{\cref{fact:SDEx}}
	\end{align*}
\end{proof}

\begin{proof}[Proof of \cref{lemma:SDVarRep}]
	\begin{align*}
	&\E_{m\getsr \rv{M}}\big[\SDC{\rv{A},\rv{B}|_{\rv{M}=m}}{ \rv{A}|_{\rv{M}=m} \times \rv{B}|_{\rv{M}=m}}\big]\\
	&=\ex{m,b\getsr \rv{M},\rv{B}}{\big[\SDC{\rv{A}|_{\rv{M}=m,\rv{B}=b}}{ \rv{A}|_{\rv{M}=m}}\big]}\tag{\cref{fact:SDEx}}\\
	&\begin{aligned}
\leq&\E_{m,b\getsr \rv{M},\rv{B}}\big[\SDC{\rv{A}|_{\rv{M}=m,\rv{B}=b}}{ \rv{A}|_{\rv{B}=b}} \\
&+\SDC{ \rv{A}|_{\rv{B}=b} }{\rv{A}}+\SDC{ \rv{A}  }{  \rv{A}|_{\rv{M}=m}}\big]
	\end{aligned}\tag{Triangle inequality}\\
	&\begin{aligned}
	=&\ex{b\getsr \rv{B}}{\SDC{\rv{A},\rv{M}|_{\rv{B}=b}}{ \rv{M}|_{\rv{B}=b}\times \rv{A}|_{\rv{B}=b} }}\\
	&+\SDC{\rv{A},\rv{B}}{\rv{B}\times \rv{A} }+\SDC{\rv{M}\times \rv{A}}{\rv{A},\rv{M}}
	\end{aligned}\tag{\cref{fact:SDEx}}\\
	&\leq2\E_{b\getsr \rv{B}}\big[\SDC{\rv{A},\rv{M}|_{\rv{B}=b}}{\rv{M}|_{\rv{B}=b}\times \rv{A}|_{\rv{B}=b}} \big]+2\SDC{\rv{A},\rv{B}}{\rv{A}\times \rv{B}}\tag{\cref{lemma:SDCondAdd}}\\
	\end{align*}
\end{proof}

\begin{proof}[Proof of \cref{claim:goodK}]

	Let $t\eqdef \ex{(\setx,\sety) \getsr (\setX,\setY)}{\size{\setx \cap \sety}}$ be the expected intersection size. We  show below that
	\begin{align}\label{eq:gapex}
	\sum_{i=0}^{\floor{4t/(1-\alpha-\gamma)}}\ppr{(\setx,\sety) \getsr (\setX, \setY)}{\size{\setx\cap\sety}=i} \cdot \Suc{i}\geq (1-\alpha-\gamma)/4
	\end{align}
	
	It will then follows that  $\exists d \le  4t/(1-\alpha-\gamma)$ such that $\Suc{d}\geq (1-\alpha-\gamma)/4$.
	We conclude the proof by showing that $t=\ell^2/\size{\cs}$, and therefore $d\le 4\ell^2/\size{\cs}(1-\alpha-\gamma)$. By linearity of expectation,
	\begin{align*}
	t& = \ex{(\setx,\sety) \getsr (\setX, \setY)}{\size{\setx\cap\sety}}= \sum_{i=1}^{\ell}\ex{(\setx,\sety) \getsr (\setX, \setY)}{\setx_i\in\sety}=\ell^2/\size{\cs}.
	\end{align*}
	
	So it is left to prove  \cref{eq:gapex}. We   first show that the expected value of $\Suc{i}$ is at least $(1-\alpha-\gamma)/2$.
	\begin{align}\label{eq:expect}
	\lefteqn{\ex{(\setx,\sety) \getsr (\setX, \setY)}{\Suc{\size{\setx\cap\sety}}}}\\
	&\begin{aligned}
	&=\ex{{(\setx,\sety) \getsr (\setX, \setY)}}{\ppr{v\getsr \Pcom(\setX,\setY)}{\out^{\Bcom}(v)=\out^{\Acom}(v)\mid \size{\x(v)\cap\y(v)}=\size{\setx\cap\sety}}}\\
	&\quad -\ex{(\setx,\sety) \getsr (\setX, \setY)}{\ppr{v\getsr \Pdist(\setX,\setY)}{\out^{\Bdist}(v)=\out^{\Adist}(v)\mid \size{\x(v)\cap\y(v)}=\size{\setx\cap\sety}}}
	\end{aligned}\nonumber\\
	&=\ppr{v\getsr \Pcom(\setX,\setY)}{\out^{\Bcom}(v)=\out^{\Acom}(v)}-\ppr{v\getsr \Pdist(\setX,\setY)}{\out^{\Bdist}(v)=\out^{\Adist}(v)\nonumber}
	\\&\geq (1-\alpha-\gamma)/2.\nonumber
	\end{align}
	It follows that 
	
	\begin{align*}
	&\sum_{i=0}^{\floor{4t/(1-\alpha-\gamma)}}\ppr{(\setx,\sety) \getsr (\setX, \setY)}{\size{\setx\cap\sety}=i} \cdot \Suc{i}\\
	&\begin{aligned}
	=&\ex{(\setx,\sety) \getsr (\setX, \setY)}{\Suc{\size{\setx\cap\sety}}}\\
	&-\sum_{i=\floor{4t/(1-\alpha-\gamma)}+1}^\ell\ppr{(\setx,\sety) \getsr (\setX, \setY)}{\size{\setx\cap\sety}=i} \cdot \Suc{i}
	\end{aligned}\\\tag{\cref{eq:expect}}
	&\geq (1-\alpha-\gamma)/2-\sum_{i=\floor{4t/(1-\alpha-\gamma)}+1}^\ell\ppr{(\setx,\sety) \getsr (\setX, \setY)}{\size{\setx\cap\sety}=i} \cdot \Suc{i}\\\tag{$\Suc{i}\leq 1$}
	&\geq (1-\alpha-\gamma)/2-\sum_{i=\floor{4t/(1-\alpha-\gamma)}+1}^\ell\ppr{(\setx,\sety) \getsr (\setX, \setY)}{\size{\setx\cap\sety}=i}\\
	&\geq (1-\alpha-\gamma)/2-\ppr{\setx\getsr \setX,\sety\getsr \setY}{\size{\setx\cap\sety}\geq 4t/(1-\alpha-\gamma)}\\
	&\geq (1-\alpha-\gamma)/4.\tag{Markov inequality},
	\end{align*}
	and the the proof of the claim follows.
\end{proof}

\begin{proof}[Proof of \cref{lemma:E0}]

	For a mapping $R : \cE_0 \rightarrow \set{0,1}^n$ representing the answers to the queries in $\cE_0$,
	let
	\begin{equation*}
	\FFam^R = \set{f \in \FFam_n\st f\big|_{\cE_0} =R} .
	\end{equation*}
	In words, it is the set of oracles whose answers on $\cE_0$ agree with $R$.
	
	We show that there is $R$ so that the protocol  $\Pi^{\FFam^R}$, where the answers to $\cE_0$ are fixed to agree with $R$, is a $(q-|\cE_0|,\alpha,\gamma)$-key agreement protocol.
	We then define $\Theta$ to be the simulation of $\Pi^{\FFam^R}$ where for each query in $\cE_0$, instead of querying the oracle the players use the answer from $R$.
	
	In $\Theta$, the queries in $\cE_0$ are never asked, so they are no longer heavy.
	Moreover, no new heavy queries are created, because 
	the protocol is non-adaptive;
	the queries $\rv{X},\rv{Y}$ asked by the players 
	do not change when we fix the answers in $\cE_0$.
	
	
	Now let us choose $R$.
	First, observe that consistency is maintained for \emph{any} setting of $R$:
	for each $f \in \FFam_n$,
	\begin{equation*}
	\ppr{v\getsr \Pi^f}{\out^\Ac(v)=\out^\Bc(v)}\geq 1-\alpha .
	\end{equation*}
	In particular this holds for $f \in \FFam^R$ for any $R$. 

	As for secrecy, assume for the sake of contradiction that there is no $R$ under which $\Pi$ is $(q-|\cE_0|,\gamma)$-secure \wrt $\FFam^R$;
	that is,
	for each $R : \cE_0 \rightarrow \set{0,1}^n$
	there exists an attacker $\Eve_R$ that asks $q-|\cE_0|$ queries
	such that 
	\begin{equation*}
	\ppr{f \getsr \FFam^R, v\getsr \Pi^f}{\Eve_R^f(trans(v))=\out^\Ac(v)}\geq \gamma.
	\end{equation*}	
	Define an attacker $\Eve$ that breaks the original protocol $\Pi$ as follows:
	First, $\Eve$ queries $\cE_0$; let $R$ be the answers she receives.
	Next, $\Eve$ simply runs $\Eve_R$.
	We have:
	\begin{align*}
	&\ppr{\rv{f} \getsr \FFam_n, v\getsr \Pi^f}{\Eve^f(trans(v))=\out^\Ac(v)}
	\geq \gamma.
	\end{align*}
	This contradicts the secrecy of $\Pi$.

\end{proof}

\begin{proof}[Proof of \cref{lemma:Y1}]

	In $\Theta$, the players execute the original protocol $\Pi$, but with the following changes:
	\begin{itemize}
		\item In the beginning of the protocol, $\Bc$ asks one additional query $\rv{Y}_{\ell+1}$.
		This query is chosen uniformly at random and independently of his other queries (and is not used by $\Pi$).
		\item $\Ac$ then sends her message $\rv{M}_1$ just as she would under $\Pi$, and $\Bc$ computes his message $\rv{M}_2$ under $\Pi$,
		and the secret key $\out^{\Bc}$ that he would output in $\Pi$.
		\item $\Bc$ sends $\Ac$ the message $\rv{M}_2, \rv{b}$,
		where $\rv{b} = \out^{\Bc} \xor (\rv{Y}_{\ell+1})_1$ is an additional bit $\Bc$ appends to the message.
		\item $\Bc$ outputs $(\rv{Y}_{\ell+1})_1$ as his secret key.
		\item $\Ac$ computes $\out^{\Ac}$ as in $\Pi$, and outputs 
		$\out^{\Ac} \xor \rv{b}$.
	\end{itemize}
	Whenever $\out^{\Ac} = \out^{\Bc}$, $\Ac$'s output agrees with $\Bc$'s.
	The consistency of the new protocol, therefore, is the same as $\Pi$'s.
	
	For secrecy, let $\rv{F}$ be the random oracle, and assume there is $\Eve^\rv{F}$ that breaks the secrecy of $\Theta$. Namely, $\Eve^\rv{F}$ can guess the output of $\Ac$ with probability at least $\gamma$. Note that $(\rv{Y}_{\ell+1})_1$ is a uniform random bit independent of $\rv{M}_1, \rv{M}_2$ and $\rv{F}$. Thus, we can think that in $\Theta$, $\Bc$ chooses the value of $\out^\Bc\oplus(\rv{Y}_{\ell+1})_1$ \emph{after} $\rv{M}_2$ was sent.
	\begin{itemize}
		\item Given a transcript $\rv{M}_1$ and $\rv{M}_2$,
		the eavesdropper $\widehat{\Eve}^\rv{F}$ chooses a uniform random bit $\rv{b}$.
		\item $\widehat{\Eve}^\rv{F}$ runs $\Eve^\rv{F}(\rv{M}_1,\rv{M}_2,\rv{b})$. Let $\out^\Eve$ be $\Eve^\rv{F}$'s output.
		\item $\widehat{\Eve}^\rv{F}$ outputs $\rv{b}\oplus \out^\Eve$.
	\end{itemize}
	Since $\rv{M}_1,\rv{M}_2,\rv{b}$ are distributed exactly as in $\Theta$, we have that $\widehat{\Eve}^\rv{F}$ breaks $\Pi$ with the same probability $\Eve^\rv{F}$ does, and with the same number of queries. 
\end{proof}

\begin{proof}[Proof of \cref{lemma:Hybrid}]

	For $z\in \rv{Z}$, let $(\rv{T'}|_z)$ be distributed as the marginal distribution of $(\rv{T}|_{Z=z})$. From the triangle inequality for statistical distance, we get:
	\begin{align*}                                                   
	&\ex{z,a_{g(z)} \getsr \rv{Z},\rv{A}_{g(z)}}{\SDC{\rv{A}_\rv{T},\rv{T},\rv{B}|_{z,a_{g(z)}}}{(\rv{A}_\rv{T},\rv{T})|_{z,a_{g(z)}}\times \rv{B}|_{z,a_{g(z)}}}}\\
	&\begin{aligned}
	\leq&\ex{z,a_{g(z)} \getsr \rv{Z},\rv{A}_{g(z)}}{\SDC{(\rv{A}_\rv{T},\rv{T},\rv{B})|_{z,a_{g(z)}}}{(\rv{A}_{\rv{T}'}, \rv{T}',\rv{B})|_{z,a_{g(z)}}}}\\
	&+\ex{z,a_{g(z)} \getsr \rv{Z},\rv{A}_{g(z)}}{\SDC{(\rv{A}_\rv{T'},\rv{T'},\rv{B})|_{z,a_{g(z)}}}{(\rv{A}_{\rv{T}'},\rv{T}')|_{z,a_{g(z)}}\times \rv{B}|_{z,a_{g(z)}}}}\\
	&+\ex{z,a_{g(z)} \getsr \rv{Z},\rv{A}_{g(z)}}{\SDC{(\rv{A}_{\rv{T}'},\rv{T}')|_{z,a_{g(z)}}\times \rv{B}|_{z,a_{g(z)}}}{(\rv{A}_\rv{T},\rv{T})|_{z,a_{g(z)}}\times \rv{B}|_{z,a_{g(z)}}}}.
	\end{aligned}	
	\end{align*}
	We bound each term above separately:
	the first term is bounded by $\delta$, because by \cref{fact:SDEx} and the data processing inequality, we have
	\begin{align*}
	&\ex{z,a_{g(z)} \getsr \rv{Z},\rv{A}_{g(z)}}{\SDC{(\rv{A}_\rv{T},\rv{T},\rv{B})|_{z,a_{g(z)}}}{(\rv{A}_{\rv{T}'}, \rv{T}',\rv{B})|_{z,a_{g(z)}}}}\\
	&=\ex{z \getsr \rv{Z}}{\SDC{\rv{A}_\rv{T},\rv{A}_{g(z)},\rv{T},\rv{B})|_z}{\rv{A}_{\rv{T}'},\rv{A}_{g(z)}, \rv{T}',\rv{B})|_{z}}}\\
	&\leq\ex{z \getsr \rv{Z}}{\SDC{(\rv{A},\rv{T},\rv{B})|_z}{(\rv{A},\rv{T}',\rv{B})|_z}}\\
	&= \ex{z \getsr \rv{Z}}{\ISD{\rv{A},\rv{B}}{\rv{T}}{z}} = \delta.
	\end{align*}
	Similarly, the third term is also bounded by $\delta$, as by data processing,
	\begin{align*}
	&\ex{z,a_{g(z)} \getsr \rv{Z},\rv{A}_{g(z)}}{\SDC{(\rv{A}_{\rv{T}'},\rv{T}')|_{z,a_{g(z)}}\times \rv{B}|_{z,a_{g(z)}}}{(\rv{A}_\rv{T},\rv{T})|_{z,a_{g(z)}}\times \rv{B}|_{z,a_{g(z)}}}}\\
	&=\ex{z,a_{g(z)} \getsr \rv{Z},\rv{A}_{g(z)}}{\SDC{(\rv{A}_{\rv{T}'},\rv{T}')|_{z,a_{g(z)}}}{(\rv{A}_\rv{T},\rv{T})|_{z,a_{g(z)}}}}\\
	&=
	\ex{z \getsr \rv{Z}}
	{\SDC{(\rv{A}_{\rv{T}'},\rv{A}_{g(z)},\rv{T}')|_{z}}{(\rv{A}_\rv{T},\rv{A}_{g(z)},\rv{T})|_{z}}}\\
	& \leq\ex{z \getsr \rv{Z}}{\ISD{\rv{A},\rv{B}}{\rv{T}}{z}}
	=\delta.
	\end{align*}
	Finally, for the second term, we can write
	\begin{align*}
	&\ex{z,a_{g(z)} \getsr \rv{Z},\rv{A}_{g(z)}}{\SDC{(\rv{A}_\rv{T'},\rv{T'},\rv{B})|_{z,a_{g(z)}}}{(\rv{A}_{\rv{T}'},\rv{T}')|_{z,a_{g(z)}}\times \rv{B}|_{a_{g(z)},z}}}\\
	&= \ex{z\getsr \rv{Z}, a_g(z)\getsr \rv{A}_g(z)}{\ex{t \getsr \rv{T}|_{z}}{\ISD{\rv{A}_{t}}{\rv{B}}{{a_{g(z)},z}}}}\tag{\cref{fact:SDEx}}\\
	&\leq\ex{z\getsr \rv{Z}, a_g(z)\getsr \rv{A}_g(z)}{\ex{t \getsr \rv{T}|_{z}}{2\sqrt{\MI(\rv{A}_t;\rv{B}|{z,a_g(z))}}}}
	\tag{\cref{fact:Pinsker}}\\
	&\leq 2\sqrt{\ex{z \getsr \rv{Z}}{\ex{t \getsr T|_{z}}{\MI(\rv{A}_t;\rv{B}|{z,\rv{A}_{g(z)})}}}}
	\tag{\cref{fact:Jensen}}\\
	& = 2\sqrt{\epsilon}.
	\end{align*}
\end{proof}

\begin{proof}[Proof of \cref{claim:deltaM1}]

	Recall that we denote by $X_{t,<i}$
	the restriction of $X$ to coordinates in $t$ that are less than $i$.
	Write
	\begin{align*}
	&
	\MI(\rv{M}_1;\rv{F}(\rv{X}_\rv{T})|\rv{T},\rv{X})
	\\
	&=
	\ex{x\getsr\rv{X}}{\ex{t\getsr \rv{T}\mid \rv{X}=x}{\MI(\rv{M}_1;\rv{F}(X_t)|\rv{T}=t,\rv{X}=x)}}
	\\
	&=
	\ex{x\getsr\rv{X}}{\ex{t\getsr \rv{T}\mid \rv{X}=x}{\sum_{i\in t}\MI(\rv{M}_1;\rv{F}(X_i)|\rv{T}=t,\rv{X}=x,\rv{F}(X_{t, <i}))} }. \tag{Chain rule}
	\end{align*}
	For fixed $x,t,i$, by the chain rule,
	\begin{align*}
	&  
	\MI(\rv{M}_1 ;\rv{F}(X_i)|\rv{T}=t,\rv{X}=x,\rv{F}(X_{t , <i})) \\
	& \leq	\MI(\rv{M}_1, \rv{F}(X_{\{1,\ldots,i-1\} \setminus t}) ;\rv{F}(X_i)|\rv{T}=t,\rv{X}=x,\rv{F}(X_{t , < i})) \\
	&\begin{aligned}
	= &\MI(\rv{F}(X_{\{1,\ldots,i-1\} \setminus t}) ;\rv{F}(X_i)|\rv{T}=t,\rv{X}=x,\rv{F}(X_{t , <i}))\\
	& +	\MI(\rv{M}_1;\rv{F}(X_i)|\rv{T}=t,\rv{X}=x,\rv{F}(X_{<i})) 
	\end{aligned}\\
	& 	= 0 + 	\MI(\rv{M}_1;\rv{F}(X_i)|\rv{T}=t,\rv{X}=x,\rv{F}(X_{<i})).
	\end{align*}
	Conditioned on $\rv{X}$, $\Ac$'s message $\rv{M}_1$ and the oracle $\rv{F}$ are independent of $\Bc$'s queries $\rv{Y}$
	and therefore also from the intersection $\rv{T}$.
	Therefore,
	
	\begin{align*}
	\MI(\rv{M}_1;\rv{F}(\rv{X}_\rv{T})|\rv{T},\rv{X})
	&	\leq 	\ex{x\getsr\rv{X}}{\ex{t\getsr \rv{T}\mid \rv{X}=x}{\sum_{i\in t}\MI(\rv{M}_1;\rv{F}(X_i)|\rv{T}=t,\rv{X}=x,\rv{F}(X_{<i}))}} \\
	&	= 	\ex{x\getsr\rv{X}}{\ex{t\getsr \rv{T}\mid \rv{X}=x}{\sum_{i\in t}\MI(\rv{M}_1;\rv{F}(X_i)|\rv{X}=x,\rv{F}(X_{<i}))}} \\
	&	= 	\ex{x\getsr\rv{X}}{\sum_i \pr{i \in T\mid \rv{X}=x} \MI(\rv{M}_1;\rv{F}(X_i)|\rv{X}=x,\rv{F}(X_{<i}))} .
	\end{align*}
	From the assumption that no queries are heavy a priori, 
	$\pr{i \in T\mid \rv{X}=x} \leq \delta$ for all $i$.
	Finally,
	\begin{align*}
	\MI(\rv{M}_1;\rv{F}(\rv{X}_\rv{T})|\rv{T},\rv{X})
	&	\leq \delta \sum_i  \MI(\rv{M}_1;\rv{F}(X_i)|\rv{X},\rv{F}(X_{<i})) \\
	&	= \delta \MI(\rv{M}_1;\rv{F}(X)|\rv{X}) \tag{Chain rule} \\
	& \leq \delta |M_1|. 	\tag{\cref{fact:IBounds}}
	\end{align*}
\end{proof}

\begin{proof}[Proof of \cref{claim:M2ReduceToIntersecton}]
	 From \cref{claim:oneMessage,cor:SDVarRep} we get that:
\begin{align*}
	&\ex{ v_E \getsr \Pi_{\Eve}^{\rv{F}} }{\ISD{\rv{X},\rv{F}(\rv{X})}{\rv{Y}}{v_E}}
	\\
	&\begin{aligned}
		\leq 2\E_{m_1,f(e_1),y\getsr\rv{M}_1, \rv{F}(\rv{\cE}_1),\rv{Y}}\SDL{\rv{X},\rv{F}(\rv{X}),\rv{M}_2|_{m_1,f(e_1),y}}{(\rv{X},\rv{F}(\rv{X})|_{m_1,f(e_1),y})\times (\rv{M}_2|_{m_1,f(e_1),y})}.
	\end{aligned}
	\end{align*}
	For every $b_E=(y,m_1,f(e_1\cap y))$,
	\begin{align*}
	&\E_{f(e_1)\getsr \rv{F}(e_1)|_{\rv{B}_E=b_E}}\SDC{\rv{X},\rv{F}(\rv{X}),\rv{M}_2|_{b_E,f(e_1)}}{(\rv{X},\rv{F}(\rv{X})\times \rv{M}_2)|_{b_E,f(e_1)}}\\
	&\leq 2\SDC{\rv{X},\rv{F}(\rv{X}),\rv{F}(e_1),\rv{M}_2|_{b_E}}{\rv{X},\rv{F}(\rv{X}),\rv{F}(e_1)|_{b_E}\times \rv{M}_2|_{b_E}}
	\tag{\cref{lemma:SDChainRule}}\\
	&= 2\ex{x,f(x)\getsr\rv{X},\rv{F}(\rv{X})|_{b_E}}{\SDC{\rv{F}(e_1),\rv{M}_2|_{b_E,x,f(x)}}{\rv{F}(e_1)|_{b_E,x,f(x)}\times \rv{M}_2|_{b_E}}}\tag{\cref{fact:SDEx}}
	\end{align*}
	Alice's message $\rv{M}_1$ is only a function of $\rv{X},\rv{F}(\rv{X}$),
	and Eve's queries $\cE_1$ are a function of $\rv{M}_1$.
	Thus, since $\rv{F}$ is a random function,
	$\rv{F}(\cE_1\setminus\rv{Y})$ is independent from $\rv{F}(\rv{Y}\setminus\cE_1)$ conditioned on $\rv{M}_1,\rv{Y}, \rv{F}(\cE_1\cap\rv{Y}), \rv{X},\rv{F}(\rv{X})$.
	Next, because $\rv{M}_2$ is a function of $\rv{M}_1$,$\rv{Y}$ and $\rv{F}(\rv{Y})$,
	we have that $\rv{M}_2$ is independent from  $\rv{F}(\cE_1\setminus\rv{Y})$ under the same conditioning.
	
	We get that the distribution $(\rv{F}(e_1),\rv{M}_2|_{b_E,x,f(x))}$ is equal to 
	\begin{align*}
	\rv{F}(e_1)|_{b_E,x,f(x)}\times \rv{M}_2|_{b_E,x,f(x)},
	\end{align*}
	and therefore,
	\begin{align*}
	&\begin{aligned} \ex{x,f(x)\getsr\rv{X},\rv{F}(\rv{X})|_{b_E}}{\SDC{\rv{F}(e_1),\rv{M}_2|_{b_E,x,f(x)}}{\rv{F}(e_1)|_{b_E,x,f(x)}\times \rv{M}_2|_{b_E}}}
	\end{aligned}\\
	&\begin{aligned} =\ex{x,f(x)\getsr\rv{X},\rv{F}(\rv{X})|_{b_E}}{\SDC{\rv{F}(e_1)|_{b_E,x,f(x)}\times \rv{M}_2|_{b_E,x,f(x)}}{\rv{F}(e_1)|_{b_E,x,f(x)}\times \rv{M}_2|_{b_E}}}
	\end{aligned}\\
	&\begin{aligned} =\ex{x,f(x)\getsr\rv{X},\rv{F}(\rv{X})|_{b_E}}{\SDC{ \rv{M}_2|_{b_E,x,f(x)}}{\rv{M}_2|_{b_E}}}
	\end{aligned}\tag{\cref{fact:SDRemoveProduct}}\\
	&\begin{aligned} =\SDC{\rv{X},\rv{F}(\rv{X}),\rv{M}_2|_{b_E}}{(\rv{X},\rv{F}(\rv{X})\times \rv{M}_2)|_{b_E}}.
	\end{aligned}\tag{\cref{fact:SDEx}}\\
	\end{align*}

	Now we can show all the dependence comes from the intersection. Since $\rv{T}_1$ is a function of $\rv{Y}$, $\rv{X}$ and $\cE_1$, and $\cE_1$ is a function of $\rv{M}_1$, we get that
	\begin{align*}
	&\begin{aligned} \SDC{\rv{X},\rv{F}(\rv{X}),\rv{M}_2|_{b_E}}{(\rv{X},\rv{F}(\rv{X})\times \rv{M}_2)|_{b_E}}
	\end{aligned}\\
	&\begin{aligned}
	=\SDC{\rv{X},\rv{T}_1,\rv{F}(y_{\rv{T}_1}),\rv{F}(\rv{X}),\rv{M}_2|_{b_E}}{
		\rv{X},\rv{T}_1,\rv{F}(y_{\rv{T}_1}),\rv{F}(\rv{X})|_{b_E}\times \rv{M}_2|_{b_E}}
	\end{aligned}\\
	&\begin{aligned}
	=\ex{t,f(y_t)\getsr\rv{T}_1,\rv{F}(y_{\rv{T}_1})|_{b_E}}{\SDC{\rv{X},\rv{F}(\rv{X}),\rv{M}_2|_{b_E,t,f(y_t)}}{
			\rv{X},\rv{F}(\rv{X})|_{b_E,t,f(y_t)}\times \rv{M}_2|_{b_E}}}
	\end{aligned}\tag{\cref{fact:SDEx}}\\
	\end{align*}
	Again, $\rv{M}_2$ is a function of $\rv{Y},\rv{F}(\rv{Y})$ and $\rv{M}_1$, and $\rv{X}, \rv{F}(\rv{X})$ are independent from $\rv{F}(\rv{Y})$ conditioned on $\rv{M}_1, \rv{Y}, \rv{F}(\cE\cap\rv{Y}), \rv{T}_1, \rv{F}(\rv{Y}_{\rv{T}_1})$.
	Thus, the distribution\\ $(\rv{X},\rv{F}(\rv{X}),\rv{M}_2|_{b_E,t,f(y_t))}$ is equal to
	\begin{align*}
	\rv{X},\rv{F}(\rv{X})|_{b_E,t,f(y_t)}\times \rv{M}_2|_{b_E,t,f(y_t)},
	\end{align*}
	and we get:
	\begin{align*}
	&\begin{aligned}
	\ex{t,f(y_t)\getsr\rv{T}_1,\rv{F}(y_{\rv{T}_1})|_{b_E}}{\SDC{\rv{X},\rv{F}(\rv{X}),\rv{M}_2|_{b_E,t,f(y_t)}}{
			\rv{X},\rv{F}(\rv{X})|_{b_E,t,f(y_t)}\times \rv{M}_2|_{b_E}}}
	\end{aligned}\\
	&\begin{aligned}
	=\E_{t,f(y_t)\getsr\rv{T}_1,\rv{F}(y_{\rv{T}_1})|_{b_E}}\big[\SDL{\rv{X},\rv{F}(\rv{X})|_{b_E,t,f(y_t)}\times \rv{M}_2|_{b_E,t,f(y_t)}}{ \rv{X},\rv{F}(\rv{X})|_{b_E,t,f(y_t)}\times \rv{M}_2|_{b_E}}\big]
	\end{aligned}\\
	&\begin{aligned}
	=\ex{t,f(y_t)\getsr\rv{T}_1,\rv{F}(y_{\rv{T}_1})|_{b_E}}{\SDC{\rv{M}_2|_{b_E,t,f(y_t)}}{ \rv{M}_2|_{b_E}}}
	\end{aligned}\tag{\cref{fact:SDRemoveProduct}}\\
	&=\SDC{\rv{T}_1,\rv{F}(y_{\rv{T}_1}),\rv{M}_2|_{b_E}}{
		(\rv{T}_1,\rv{F}(y_{\rv{T}_1})\times \rv{M}_2)|_{b_E}}.\tag{\cref{fact:SDEx}}
	\end{align*}
	
	To conclude the proof, we take the expectation over $\rv{B}_E$, and the claim follows by the monotonicity of expectation.
\end{proof}

\begin{proof}[Proof of \cref{claim:deltaM2}]
	\begin{align*}
	&
	\E_{m_1,y\getsr \rv{M}_1,\rv{Y}}\E_{t\getsr \rv{T}_1|_{m_1,y}}[\MI(\rv{F}(y_t);\rv{M}_2|m_1,y,\rv{F}(e_1\cap y))]
	\\
	&=
	\E_{m_1,y\getsr \rv{M}_1,\rv{Y}}\E_{t\getsr \rv{T}_1|_{m_1,y}}\bigg[\sum_{i \in t}\MI(\rv{F}(y_i);\rv{M}_2|m_1,y,\rv{F}(e_1\cap y),\rv{F}(y_{t, <i}))\bigg]\tag{Chain rule}
	\\
	&\begin{aligned}
	=
	\E_{m_1,y\getsr \rv{M}_1,\rv{Y}}&\ex{t\getsr \rv{T}_1|_{m_1,y}}{\sum_{i \in t}\MI(\rv{F}(y_i);\rv{M}_2|m_1,y, \rv{F}(e_1\cap y),\rv{F}(y_{<i}))}\\
	& + \E_{m_1,y\getsr \rv{M}_1,\rv{Y}}\E_{t\getsr \rv{T}_1|_{m_1,y}}\bigg[\sum_{i \in t}\MI(\rv{F}(y_i);\rv{M}_2|m_1,y, \rv{F}(e_1\cap y),\rv{F}(y_{t,<i}))\\
	&-\MI(\rv{F}(y_i);\rv{M}_2|m_1,y, \rv{F}(e_1\cap y),\rv{F}(y_{<i}))\bigg]
	\end{aligned}
	\\
	&
	\begin{aligned}
	\leq
	\E_{m_1,y\getsr \rv{M}_1,\rv{Y}}&\E_{t\getsr \rv{T}_1|_{m_1,y}}\bigg[\sum_{i \in t}\MI(\rv{F}(y_i);\rv{M}_2|m_1,y,\rv{F}(e_1\cap y),\rv{F}(y_{<i}))\bigg]\\
	&+\E_{m_1,y\getsr \rv{M}_1,\rv{Y}}\E_{t\getsr \rv{T}_1|_{m_1,y}}\bigg[\sum_{i \in t}\MI(\rv{F}(y_i);\rv{F}(y_{<i})|m_1,y, \rv{F}(e_1\cap y))\bigg]
	\end{aligned}\tag{\cref{lemma:ICondAdd}}
	\\
	&\begin{aligned}
	=
	\E_{m_1,y\getsr \rv{M}_1,\rv{Y}}&\E_{t\getsr \rv{T}_1|_{m_1,y}}\sum_{i \in t}\bigg[\MI(\rv{F}(y_i);\rv{M}_2|m_1,y, \rv{F}(e_1\cap y),\rv{F}(y_{<i}))\\
	&+\MI(\rv{F}(y_i);\rv{F}(y_{<i})|m_1,y, \rv{F}(e_1\cap y))\bigg]
	\end{aligned}\\
	&\begin{aligned}
	=
	\E_{m_1,y\getsr \rv{M}_1,\rv{Y}}&\sum_{i\in [\ell]}\pr{i \in \rv{T}_1|{m_1,y}}\bigg[\MI(\rv{F}(y_i);\rv{M}_2|m_1,y, \rv{F}(e_1\cap y),\rv{F}(y_{<i}))\\
	&+\MI(\rv{F}(y_i);\rv{F}(y_{<i})|m_1,y, \rv{F}(e_1\cap y))\bigg]
	\end{aligned}
	\end{align*}
	Since we excluded the heavy queries $\cE_1$ from $\rv{T}_1$, and $y_i$ is some fixed query, and since $\rv{X}$ is independent from $\rv{Y}$ conditioned on $\rv{M}_1$ we have 
	\begin{align*}
	\pr{i\in \rv{T}_1|{m_1,y}} = \pr{y_i\in (\rv{X} \setminus \cE_1)|{m_1,y}}\leq\pr{y_i\in (\rv{X} \setminus \cE_1)|{m_1}}\\
	\leq \pr{y_i\in ((\rv{X}\cup \rv{Y}) \setminus \cE_1)|{m_1}}\leq \delta.
	\end{align*}
	Therefore,
	\begin{align*}
	&\begin{aligned}\E_{m_1,y\getsr \rv{M}_1,\rv{Y}}&\sum_{i\in [\ell]}\pr{i \in \rv{T}_1|_{m_1,y}}\bigg[\MI(\rv{F}(y_i);\rv{M}_2|m_1,y, \rv{F}(e_1\cap y),\rv{F}(y_{<i}))\\
	&+\MI(\rv{F}(y_i);\rv{F}(y_{<i})|m_1,y, \rv{F}(e_1\cap y))\bigg]
	\end{aligned}\\
	&\begin{aligned}
	\leq
	\E_{m_1,y\getsr \rv{M}_1,\rv{Y}}&\sum_{i}\delta\bigg[\MI(\rv{F}(y_i);\rv{M}_2|m_1,y, \rv{F}(e_1\cap y),\rv{F}(y_{<i}))\\
	&+\MI(\rv{F}(y_i);\rv{F}(y_{<i})|m_1,y, \rv{F}(e_1\cap y))\bigg]
	\end{aligned}
	\\
	&\begin{aligned}
	\leq
	\delta\E_{y\getsr \rv{Y}}\bigg[&\MI(\rv{F}(y);\rv{M}_2|\rv{M}_1,y,\rv{F}(\rv{\cE}_1\cap y))\\
	&+\sum_{i}\MI(\rv{F}(y_i);\rv{F}(y_{<i})|\rv{M}_1,y, \rv{F}(\rv{\cE}_1\cap y))\bigg]		
	\end{aligned}\tag{Chain rule}
	\\
	&\leq\delta\E_{y\getsr \rv{Y}}\bigg[ |\rv{M}_2|+\sum_{i}\MI(\rv{F}(y_i);\rv{F}(y_{<i})|\rv{M}_1,y,\rv{F}(\rv{\cE}_1\cap y))\bigg]\tag{\cref{fact:IBounds}}\\
	\end{align*}	
\end{proof}

\end{document}